\documentclass{article}
\usepackage[hang,flushmargin]{footmisc}
\usepackage[margin=1in]{geometry}
\usepackage{titling}
\usepackage{float}
\usepackage{array}
\usepackage{microtype}
\usepackage{amsthm,amsmath}
\usepackage{mathtools}
\usepackage[hypcap=false]{caption}
\usepackage{subcaption}
\usepackage{amsfonts}
\usepackage{setspace}
\usepackage{xcolor}
\usepackage{nicefrac}
\usepackage[boxruled,
  linesnumbered,
  commentsnumbered,
  procnumbered,
  ]{algorithm2e}
\usepackage{hyperref}
\usepackage{cleveref}

\makeatletter\let\expandableinput\@@input\makeatother

\RequirePackage[authoryear]{natbib}
\usepackage[utf8]{inputenc}
\usepackage{graphics}

\theoremstyle{plain}

\newtheorem{proposition}{Proposition}[section]

\theoremstyle{definition}
\newtheorem{definition}{Definition}[section]

\newtheorem{prf}{Proof}[section]

\theoremstyle{remark}

\usepackage{etoolbox}

\newtoggle{journalorarxiv}
\newcommand{\ja}[2]{\iftoggle{journalorarxiv}{#1}{#2}}

\newtoggle{withfigs}
\newcommand{\wf}[2]{\iftoggle{withfigs}{#1}{#2}}

\DontPrintSemicolon

\makeatletter
\renewcommand{\SetKwInOut}[2]{\sbox\algocf@inoutbox{\KwSty{#2}\algocf@typo:}\expandafter\ifx\csname InOutSizeDefined\endcsname\relax \newcommand\InOutSizeDefined{}\setlength{\inoutsize}{\wd\algocf@inoutbox}\sbox\algocf@inoutbox{\parbox[t]{\inoutsize}{\KwSty{#2}\algocf@typo:\hfill}~}\setlength{\inoutindent}{\wd\algocf@inoutbox}\else \ifdim\wd\algocf@inoutbox>\inoutsize \setlength{\inoutsize}{\wd\algocf@inoutbox}\sbox\algocf@inoutbox{\parbox[t]{\inoutsize}{\KwSty{#2}\algocf@typo:\hfill}~}\setlength{\inoutindent}{\wd\algocf@inoutbox}\fi \fi \algocf@newcommand{#1}[1]{\ifthenelse{\boolean{algocf@inoutnumbered}}{\relax}{\everypar={\relax}}{\let\\\algocf@newinout\hangindent=\inoutindent\hangafter=1\parbox[t]{\inoutsize}{\KwSty{#2}\algocf@typo:\hfill}~##1\par}\algocf@linesnumbered }}\makeatother

\SetKwInOut{Input}{Input}
\SetKwInOut{Output}{Output}
\SetKw{In}{in}
\SetKwProg{Function}{function}{:}{}

\SetCommentSty{commfont}
\SetKwComment{Comm}{$\rhd\ $}{}

\newcommand\bb[1]{\mathbb{#1}}
\newcommand\ca[1]{\mathcal{#1}}

\newcommand\green[1]{#1}

\newcounter{noteWUctr} 

\newcounter{noteAEctr} 

\newcounter{noteMCctr} 
 
\togglefalse{journalorarxiv}

\toggletrue{withfigs}

\title{Motif-Based Spectral Clustering of
  Weighted Directed Networks}

\makeatletter
\def\nmfootnote{\gdef\@thefnmark{}\@footnotetext}
\makeatother

\author{
William G.~Underwood\textsuperscript{1,2*}
\and
Andrew Elliott\textsuperscript{3,4}
\and
Mihai Cucuringu\textsuperscript{1,3}
}

\makeatletter
\begin{document}

\setlength{\abovedisplayskip}{-3pt}
\setlength{\belowdisplayskip}{10pt}
\setlength{\abovedisplayshortskip}{-6pt}

\maketitle

\nmfootnote{\textsuperscript{*}Correspondence:
  \href{mailto:wgu2@princeton.edu}{\texttt{wgu2@princeton.edu}}}
\footnotetext[1]{Department of Statistics,
  University of Oxford,
  24--29 St. Giles,
  Oxford,
  OX1 3LB,
  UK}
\footnotetext[2]{Department of Operations Research
  and Financial Engineering,
  Princeton University,
  Sherrerd Hall,
  Charlton Street,
  Princeton,
  NJ 08544,
  USA
}
\footnotetext[3]{The Alan Turing Institute,
  British Library,
  96 Euston Road,
  London,
  NW1 2DB,
  UK
}
\footnotetext[4]{School of Mathematics and Statistics,
  University of Glasgow,
  University Place,
  Glasgow,
  GL12 8QQ,
  UK
}
\setcounter{footnote}{4}

\begin{abstract}
  \noindent Clustering is an essential technique for network analysis, with applications in
a diverse range of fields. Although spectral clustering is a popular and
effective method, it fails to consider higher-order structure and can perform
poorly on directed networks. One approach is to capture and cluster higher-order structures using motif adjacency matrices. However, current formulations fail to take edge weights into account,  and thus are somewhat  limited when weight is a key component of the  network under study.

We address these shortcomings by exploring motif-based weighted spectral clustering methods. We present new and computationally useful matrix formulae for motif adjacency matrices on weighted networks, which can be used to construct efficient algorithms for any anchored or non-anchored motif on three nodes. In a very sparse regime, our proposed method can handle graphs with \green{a million nodes} and tens of millions of edges. We further use our framework to construct a motif-based approach for clustering bipartite networks.

We provide comprehensive experimental results, demonstrating (i) the scalability
of our approach, (ii) advantages of higher-order clustering on synthetic
examples, and (iii) the effectiveness of our techniques on a variety of real
world data sets;
\green{and compare against several techniques from the literature}.
We conclude that motif-based spectral clustering is a valuable
tool for analysis of directed and bipartite weighted networks, which is also scalable and easy to implement. \end{abstract}

\section{Introduction}

Networks are ubiquitous in modern society;
from the internet and online blogs
to protein interactions and human migration,
we are surrounded by inherently
connected structures~\citep{kolaczyk2014statistical}.
The mathematical and
statistical analysis of networks is therefore an important area of modern research,
with applications in a diverse range of fields, including biology~\citep{albert2005scale}, chemistry~\citep{jacob2018statistics}, physics~\citep{newman2008physics} and sociology~\citep{adamic2005political}.

A common task in network analysis is that of
\emph{clustering}~\citep{schaeffer2007graph}.
Network clustering refers to the
partitioning
of a network into ``clusters'',
so that nodes in each cluster are similar (in some sense),
while nodes in different clusters are dissimilar.
For a review of approaches, see for example
\citep{fortunato2010community}
or
\citep{fortunato2016community}.
Spectral methods for network clustering have a long and successful
history~\citep{cheeger1969lower,donath1972algorithms,guattery1995performance},
and have become increasingly popular in recent years.  These techniques exhibit
many attractive properties including generality, ease of implementation and scalability~\citep{von2007tutorial};
in addition to often being amenable to theoretical analysis by using tools from matrix perturbation theory~\citep{stewart1990matrix}.
However, traditional spectral methods have shortcomings, particularly involving
their inability to consider higher-order network structures
\green{(organizations above the level of individual nodes and edges),
which have become of increasing interest in recent years \citep{rosvall2014memory,benson2016higher,benson2018simplicial}};
and their insensitivity to edge
directions\footnote{Although there are some spectral approaches
which do consider edge direction e.g. \cite{rohe2016co,satuluri2011sym}.}
\citep{DirectedClustImbCuts}.
Such weaknesses can lead to unsatisfactory results, especially when considering directed networks.
Motif-based spectral methods have proven more effective for clustering directed
networks on the basis of higher-order
structures~\citep{tsourakakis2017scalable},
with the introduction of the
\emph{motif adjacency matrix} (MAM) \citep{benson2016higher}.
\green{While weights can be important in network clustering~\citep{newman2004analysis},
to the best of our knowledge,
these motif-based methods have not been
comprehensively investigated
on weighted networks.
Thus, in this paper, we focus on extending these methods to
the family of
weighted directed networks.}

\paragraph{Contribution}
In this paper, we explore motif-based spectral clustering methods with a
focus on
addressing these shortcomings by
generalizing motif-based spectral methods to
\emph{weighted} directed networks.
Our main contributions include a collection of new matrix-based formulae for MAMs
on weighted directed networks,
and a motif-based approach for clustering
bipartite networks.
We also provide
computational analysis of our approaches,
demonstrations of scalability and
comprehensive experimental results, both
from synthetic data (variants of stochastic block models)
and from real world network data.
\green{Finally, we provide a thoroughly tested, scalable implementation of our proposed matrix-based MAM formulae in both Python and R, which can be found at \url{https://github.com/wgunderwood/motifcluster}.}

\paragraph{Paper layout}
Our paper is organized as follows.
In \Cref{chap:graphs}, we describe our graph-theoretic framework which provides a natural model for real world weighted directed networks and weighted bipartite networks.
In  \Cref{sec:methodology}
we develop our methodology,
and state and prove new matrix-based formulae for MAMs. We explore the computational complexity of our approaches
and demonstrate their scalability on sparse graphs.
In \Cref{sec:motif_dsbms}
we explore the performance of our approaches on
several synthetic examples.
We demonstrate
the utility of considering weight and higher-order structure,
and compare against non-weighted and non-higher-order methods.
In  \Cref{sec:realWorld}, we
apply our methods to
real world data sets,
demonstrating that they can uncover
interesting divisions in
\green{weighted}
directed
\green{networks}
and
\green{in weighted}
bipartite
networks.
\green{We compare our performance to standard methods and highlight
our ability}
to avoid misclassification.
Finally, in \Cref{chap:conclusions}
we present our conclusions
and discuss future work. \section{Framework}
\label{chap:graphs}

In this section, we give notation and definitions for
our graph-theoretic framework,
with the aim of being able to define our weighted
generalizations of motif adjacency matrices in
\Cref{sec:methodology}.
The motif adjacency matrix
\citep{benson2016higher}
is the central object in
motif-based spectral clustering, and
serves as a similarity matrix for spectral clustering.
In an unweighted MAM $M$, the entry $M_{ij}$ is proportional to
the number of motifs of a given type that include both of the vertices $i$ and $j$.

\paragraph{Notation}
Graph notation is notoriously inconsistent in the literature;
in this work,
a \emph{graph} is a triple $\ca{G} = (\ca{V,E},W)$ where $\ca{V}$ is the \emph{vertex set}, $\ca{E} \subseteq \left\{ (i,j) : i,j \in \ca{V}, i \neq j \right\}$ is the \emph{edge set} and $W\colon \ca{E} \to (0,\infty)$ is the \emph{weight map}.
A graph $\ca{G'} = (\ca{V',E'})$ is a \emph{subgraph} of a graph $\ca{G} = (\ca{V,E})$ (write $\ca{G'} \leq \ca{G}$) if $\ca{V'} \subseteq \ca{V}$ and $\ca{E'} \subseteq \ca{E}$.
It is an \emph{induced subgraph} (write $\ca{G'} < \ca{G}$) if further $\ca{E'} = \ca{E} \cap ( \ca{V'} \times \ca{V'} )$.
A graph $\ca{G'} = (\ca{V',E'})$ is \emph{isomorphic} to a graph $\ca{G} = (\ca{V,E})$ (write $\ca{G'} \cong \ca{G}$) if there exists a bijection $\phi\colon \ca{V'} \rightarrow \ca{V}$ with $(u,v) \in \ca{E'} \iff \big(\phi(u), \phi(v) \big) \in \ca{E}$.
An isomorphism from a graph to itself is called an \emph{automorphism}.
Where it is not relevant or understood from the context, we may sometimes omit the weight map $W$.

\ja{}{\pagebreak}
As we are considering directed weighted graphs,
it is convenient to consider five indicator matrices that capture the different possible relationships between pairs of nodes, namely
the directed
indicator matrix $J$,
the single-edge indicator matrix $J_\mathrm{s}$,
the double-edge indicator matrix $J_\mathrm{d}$,
the missing-edge indicator matrix $J_0$,
and the vertex-distinct indicator matrix
$J_\mathrm{n}$:

\begin{align*}
	J_{ij} &\vcentcolon= \bb{I} \{ (i,j) \in \ca{E} \}\,, \\
	(J_\mathrm{s})_{ij} &\vcentcolon= \bb{I} \{ (i,j) \in \ca{E} \textrm{ and } (j,i) \notin \ca{E} \}\,, \\
	(J_\mathrm{d})_{ij} &\vcentcolon= \bb{I} \{ (i,j) \in \ca{E} \textrm{ and } (j,i) \in \ca{E} \}\,, \\
	(J_0)_{ij} &\vcentcolon= \bb{I} \{ (i,j) \notin \ca{E} \textrm{ and } (j,i) \notin \ca{E} \textrm{ and } i \neq j \}\,, \\
	(J_\mathrm{n})_{ij} &\vcentcolon= \bb{I} \{ i \neq j \}\,,
\end{align*}

\noindent where $\bb{I}$ is an indicator function.
Furthermore, we consider the following three weighted adjacency matrices,
corresponding to directed edges, single edges, and double edges, respectively:

\begin{align*}
	G_{ij} &\vcentcolon= W((i,j)) \ \bb{I} \{ (i,j) \in \ca{E} \}\,, \\
	(G_\mathrm{s})_{ij} &\vcentcolon= W((i,j)) \ \bb{I} \{ (i,j) \in \ca{E} \textrm{ and } (j,i) \notin \ca{E} \}\,, \\
	(G_\mathrm{d})_{ij} &\vcentcolon= \big( W((i,j)) + W((j,i)) \big) \ \bb{I} \{ (i,j) \in \ca{E} \textrm{ and } (j,i) \in \ca{E} \}\,.
\end{align*}

\noindent We can extend to the setting of undirected graphs by constructing a new graph, where each undirected edge is replaced by a bi-directional edge.

\begin{definition}[Motifs and anchor sets]
A \emph{motif} is a pair $(\ca{M,A})$ where $\ca{M} = (\ca{V_M,E_M})$ is a
(weakly) connected graph with $\ca{V_M} = \{ 1, \ldots, m \}$ for some small $m \geq 2$, and
an \emph{anchor set} is
$\ca{A} \subseteq \ca{V_M}$ with $|\ca{A}| \geq 2$.
If $\ca{A} \neq \ca{V_M}$, we say the motif is \emph{anchored}, and if $\ca{A=V_M}$ we say it is \emph{simple}.
We say that $\ca{H}$ is a \emph{functional instance} of $\ca{M}$ in $\ca{G}$ if $\ca{M} \cong \ca{H} \leq \ca{G}$,
and we say that $\ca{H}$ is a \emph{structural instance} of $\ca{M}$ in $\ca{G}$ if $\ca{M} \cong \ca{H} < \ca{G}$.
When an anchor set is not given, it is assumed that the motif is simple.
\Cref{fig:motif_definitions_directed} shows all the simple motifs (up to isomorphism) on at most three vertices.
\end{definition}

\wf{\begin{figure}[t]
	\centering
	\includegraphics[width=0.8\textwidth]{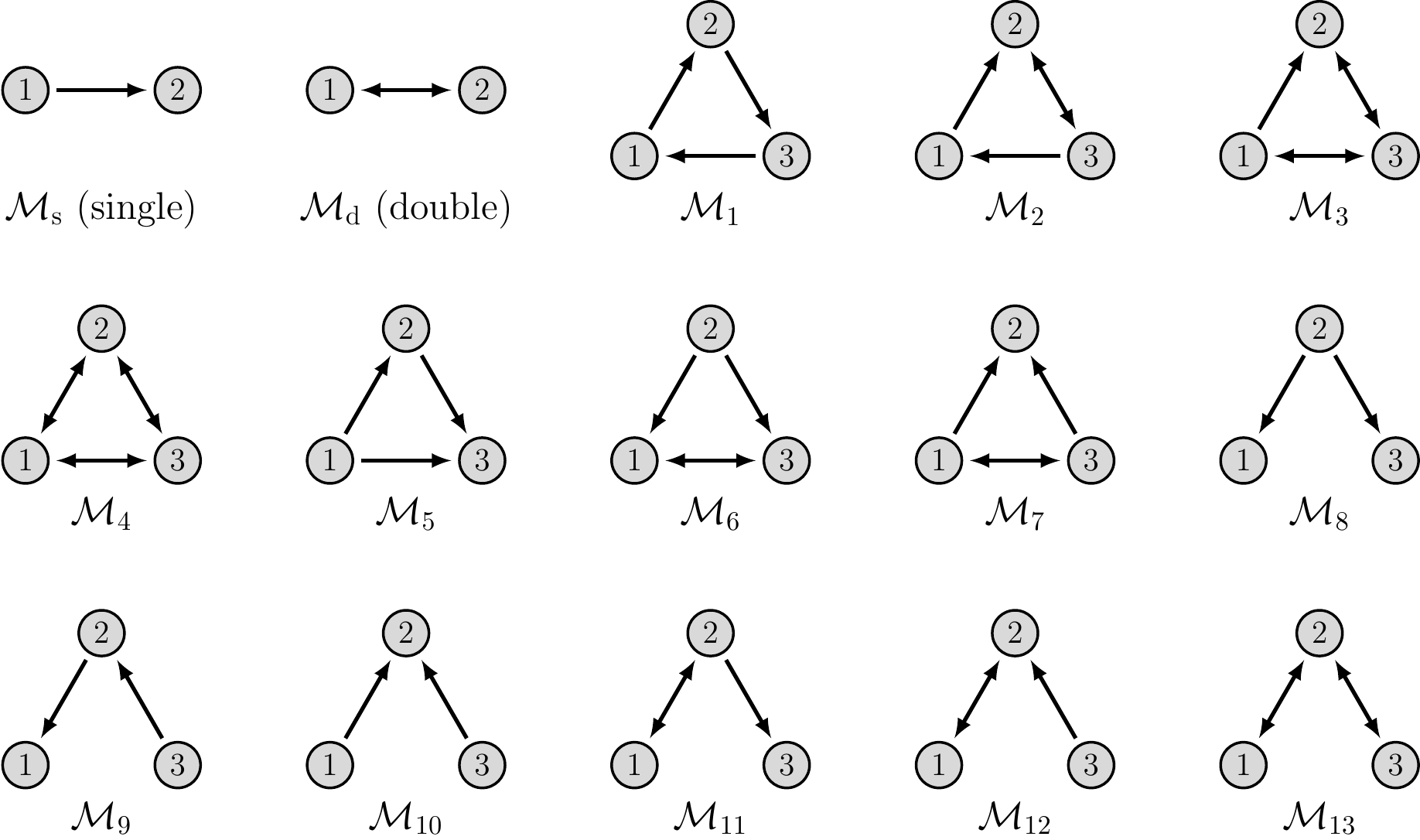}
	\caption{All simple directed motifs on at most three vertices.}
	\label{fig:motif_definitions_directed}
\end{figure}
}{}

\noindent Structural instances occur when an exact copy of the motif is
present in the graph:
i.e. edges
present (resp. not present) in the motif are
present
(resp. not present) in the graph.
Functional instances are less restrictive,
and occur when the motif is present in a graph,
but potentially with extra edges.

Anchor sets~\citep{benson2016higher}
can be thought of as the set of locations in a motif that we consider important
for the motif structure.
For example we could consider the 2-path motif
$\ca{M}_9$,
and try to cluster together nodes which appear
at the start or end of a 2-path, but maybe not in the middle.
Then we can define the anchor set as the subset of
motif vertices which should not be separated.
In $\ca{M}_9$,
the start and end vertices would correspond to
$\ca{A} = \{1,3\}$
(\Cref{fig:motif_definitions_directed}).
Anchor sets are
crucial for defining the
collider and
expander motifs given in \Cref{sec:bipartite}.

Finally, we require one additional definition before we can state our generalization of MAMs to weighted networks.

\begin{definition}[Anchored pairs]
Let $\ca{G}$ be a graph and $(\ca{M,A})$ a motif. Suppose $\ca{H}$ is an instance of $\ca{M}$ in $\ca{G}$. Define the \emph{anchored pairs of the instance} $\ca{H}$ as

\[ \ca{A(H)} \vcentcolon = \big\{ \{\phi(i),\phi(j)\} : i,j \in \ca{A}, \ i \neq j, \ \phi \textrm{ is an isomorphism from } \ca{M} \textrm{ to } \ca{H} \big\}\,.\]

\noindent This is the set of pairs of vertices
for which both vertices
lie in the image of the motif's anchor set,
under some isomorphism from the motif
to the instance.
\end{definition}

\section{Methodology} \label{sec:methodology}

In this section
we detail our methods for motif-based spectral
clustering of weighted directed networks.
Firstly we define our weighted generalizations of motif adjacency matrices (MAMs)
(\Cref{def:motif_adj_matrices}).
We further provide computationally useful formulae for weighted MAMs
(\Cref{prop:motif_adj_matrix_formula}),
and discuss their applications to clustering
(\Cref{sec:clustering_mams}).
Finally we present
a complexity analysis of our method
for computing weighted MAMs
(\Cref{prop:motif_adj_matrix_computation}),
and empirically demonstrate the scalability of our approach
(\Cref{sec:computational_analysis}).

\wf{\begin{figure}[t]
  \centering
  \includegraphics[scale=0.8]{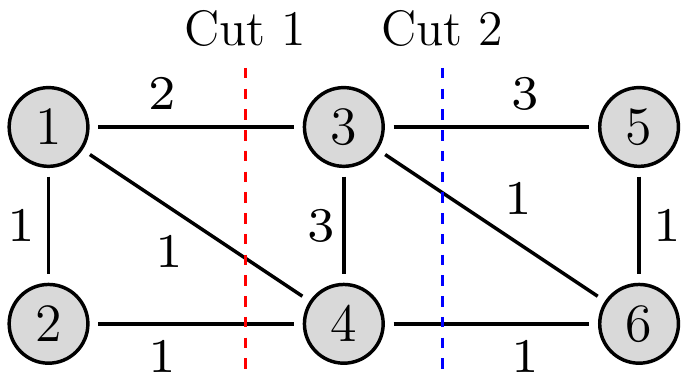}
  \caption{\green{An example illustrating that different weighting schemes
      prefer different motif cuts: the four triangles have mean-weights of 2, 1,
      $\frac{5}{3}$ and $\frac{5}{3}$, and product-weights of 6, 1, 3 and 3
      respectively whereas the minimum edge weight in each triangle is $1$. Thus, the minimum formulation cannot distinguish between the cuts,
      whereas, under the mean formulation Cut 1 gives a cut of size $9$, and Cut 2 gives a cut of size $10$. On the other hand, under the product formulation, Cut 1 gives a cut of size $7$, and Cut 2 gives a cut of size $6$. Hence the mean formulation prefers Cut 1 and product formulation prefers Cut 2.}}
\label{fig:weighted_cuts}
\end{figure}
}{}

\subsection{Weighted motif adjacency matrices}\label{sec:weighted_mams}

\green{Generalizing unweighted measures to weighted networks is non-trivial
and application-dependent,
and there are typically many possible valid choices.
For example, see the approaches for motif weighting proposed by
\cite{onnela2005intensity}
and
\cite{benson2018simplicial}.}
\green{It is often helpful to first consider
the generalization of the measure
to multi-edges \citep{newman2004analysis}.
We can then view
positive integer-weighted edges as multi-edges,
and extend to positive real-weighted edges in a natural way.}

\green{Thus, we begin by considering the
weight to place
on a motif which contains multi-edges.}
\green{A first option might be to consider the minimum number
of multi-edges lying on any edge of the motif.
This would capture the number of fully
edge-disjoint instances of the motif.}
\green{Another option would be to count the number of
unique instances of the motif
(possibly counting individual edges more than once),
giving the total weight as the product of
the number of multi-edges between each pair of nodes
in the instance.}
\green{Alternatively,
as a compromise between these two schemes,
we might consider the number of distinct motifs present,
were the multi-edges distributed evenly
among all the edges in the motif
(allowing fractional edges for simplicity).
This gives the arithmetic mean weighting scheme which is used in
\cite{onnela2005intensity},
\cite{benson2018simplicial},
\cite{mora2018pymfinder}
and
\cite{bmotif_r}.}

\green{As an illustrative example
of how this choice can affect clustering,
\Cref{fig:weighted_cuts} shows how
different motif weighting schemes prefer to divide a network in different ways.
Here, we treat the undirected edges as
bi-directional edges and consider the
fully-connected triangle motif $\ca{M}_4$.
The ``minimum'' approach has no preference between Cut 1 and Cut 2,
with all motifs having a weight of 1, and in this example is equivalent to the simple unweighted case.
However the ``mean'' weighting approach
(which, up to a scaling factor,
is equivalent to summing the edge weights)
prefers Cut 1
(cutting motifs of mean-weight $2$ and $1$ rather
than $\frac{5}{3}$ and $\frac{5}{3}$),
while the ``product'' approach prefers Cut 2
(cutting motifs of product-weight 3 and 3 rather than 6 and 1).
Further,
although they ostensibly concern bipartite networks,
\Cref{sec:example3} exhibits some of the
effects of different weighting schemes
on the performance of motif-based clustering,
and \Cref{sec:languages} provides real world motivation
for using mean-weighted motifs.}

\green{Each of these approaches has merits.
The ``minimum'' approach might be natural when
dealing with flow networks,
where low-count multi-edges could indicate bottlenecks
or points of unreliability.
The ``product'' approach is useful if we want to consider each
possible set of edges to be a separate entity,
while the ``mean'' approach is appropriate if we consider
the motif to be a single object
and wish to count the total number of edges it contains,
perhaps as some notion of capacity.}

\green{When considering weighted edges however,
these approaches have some mathematical differences,
particularly in how they handle the presence of large weights.
For example, suppose that a graph contains
just a few heavy edges.
Were two of these heavy edges to appear together in a motif instance,
the product formulation would assign a very large total weight
to that instance,
possibly completely dominating any other instances.
On the other hand,
using the mean formulation ensures
that the total weight of any instance remains on the same scale
as the individual edge weights.
By taking the minimum edge weight,
it may be that the heavy edges do not contribute to
any motif weights at all.}

\green{With these points in mind,
and following
\cite{mora2018pymfinder}
and
\cite{bmotif_r},
we use the ``mean-weighted'' approach in this paper,
defining
the weight of a
motif instance by its average (mean) edge weight.
We acknowledge that other weighting schemes may be more appropriate in some circumstances}\footnote{For convenience our software package implements three weighting schemes: product, mean and a scheme which ignores the weights.}.
\green{For example,
\cite{onnela2005intensity} introduces several weighting schemes related to  the geometric and arithmetic means of the edge weights.
Further,
these
are special cases of
the generalized $p$-means weighting scheme detailed by
\cite{benson2018simplicial}.}

Now that we have chosen a motif weighting scheme,
we can define our weighted generalizations
of motif adjacency matrices.
We note that while a weighted extension was considered by \cite{wang2018weighted}, in this paper we provide the first thorough exploration, including computational analysis, fast software implementations and comprehensive real world and synthetic experiments.

\begin{definition}[Weighted motif adjacency matrices]
\label{def:motif_adj_matrices}
Let $\ca{G} = (\ca{V,E},W)$ be a weighted graph on $n$ vertices and let $\ca{(M,A)}$ be a motif.
The \emph{functional} and \emph{structural}
\green{(mean-weighted)}
\emph{motif adjacency matrices} (MAMs) of size $n \times n$ of $\ca{(M,A)}$ in $\ca{G}$ are  given by

\begin{align*}
M^\mathrm{func}_{ij} &\vcentcolon= \frac{1}{|\ca{E_M}|} \sum_{\ca{M} \cong \ca{H} \leq \ca{G}} \bb{I} \big\{ \{i,j\} \in \ca{A}(\ca{H}) \big\} \sum_{e \in \ca{E_H}} W(e)\,, \\
M^\mathrm{struc}_{ij} &\vcentcolon= \frac{1}{|\ca{E_M}|} \sum_{\ca{M} \cong \ca{H} < \ca{G}} \bb{I} \big\{ \{i,j\} \in \ca{A}(\ca{H}) \big\} \sum_{e \in \ca{E_H}} W(e)\,.
\end{align*}
\end{definition}

\noindent When $W \equiv 1$ and $\ca{M}$ is simple,
the (functional or structural) MAM entry
$M_{ij}$ $(i \neq j)$ simply counts the (functional or structural) instances of $\ca{M}$ in $\ca{G}$ containing $i$ and $j$.
When $\ca{M}$ is not simple, $M_{ij}$ counts only those instances with anchor sets containing both $i$ and $j$.
MAMs are always symmetric, since the only dependency on $(i,j)$ is via the unordered set $\{i,j\}$.
In order to state \Cref{prop:motif_adj_matrix_formula},
we need one more definition.

\begin{definition}[Anchored automorphism classes]
Let $(\ca{M,A})$ be a motif.
Let $S_\ca{M}$ be the set of permutations on $ \ca{V_M} = \{ 1, \ldots, m \}$ and define the \emph{anchor-preserving permutations} $S_\ca{M,A} = \{ \sigma \in S_\ca{M} : \{1,m\} \subseteq \sigma(\ca{A}) \}$.
Let $\sim$ be the equivalence relation defined on $S_\ca{M,A}$ by: $\sigma \sim \tau \iff \tau^{-1} \sigma$ is an automorphism of $\ca{M}$.
Finally the \emph{anchored automorphism classes} are the quotient set
$S_\ca{M,A}^\sim \vcentcolon= S_\ca{M,A} \ \big/ \sim$\,.
\end{definition}

\noindent The motivation for this definition of
anchored automorphism classes is as follows:
suppose we are looking for instances of
$(\ca{M},\ca{A})$
in $\ca{G}$
which contain nodes $i$ and $j$.
We set $k_1=i$ and $k_m=j$
(where $m$ is the number of nodes in the motif --
note we could use any two fixed indices
instead of $1$ and $m$ here),
and choose some other nodes
$\{k_2, \ldots, k_{m-1}\}$.
We want to find all mappings of vertices in
$\ca{M}$
to our chosen vertices in
$\ca{G}$
by
$u \mapsto k_{\sigma u}$
such that
(i) vertices in $\ca{A}$ are mapped to
$i$ and $j$,
and
(ii) mappings which correspond to the same instance are not counted more than once.
These conditions are precisely the same as requiring
$\sigma \in S_\ca{M,A}^\sim$\,.

\begin{proposition}[MAM formula] \label{prop:motif_adj_matrix_formula}

Let $\ca{G} = (\ca{V,E},W)$ be a graph with vertex set ${\ca{V}=\{1,\ldots,n\}}$ and let $(\ca{M,A})$ be a motif on $m$ vertices. For any $i,j \in \ca{V}$ and with $k_1 = i$, $k_m = j$, the functional and structural MAMs of $\ca{(M,A)}$ in $\ca{G}$ are given by

\vspace{2mm}
\begin{align}
M^\mathrm{func}_{ij} &= \frac{1}{|\ca{E_M}|} \sum_{\sigma \in S_\ca{M,A}^\sim} \ \sum_{\{k_2, \ldots, k_{m-1}\} \subseteq \ca{V}} \ J^\mathrm{func}_{\mathbf{k},\sigma} \ G^\mathrm{func}_{\mathbf{k},\sigma}\,, \label{eq:mamformulafunc} \\
M^\mathrm{struc}_{ij} &= \frac{1}{|\ca{E_M}|} \sum_{\sigma \in S_\ca{M,A}^\sim} \ \sum_{\{k_2, \ldots, k_{m-1}\} \subseteq \ca{V}} \ J^\mathrm{struc}_{\mathbf{k},\sigma} \ G^\mathrm{struc}_{\mathbf{k},\sigma}\,, \label{eq:mamformulastruc}
\end{align}
\ja{\vspace{1mm}}{}

\noindent  where
$J^\mathrm{func}_{\mathbf{k},\sigma}$
is equal to one
(and zero otherwise)
if
$\ca{M}$
appears as a functional
(likewise for structural)
instance on the $m$-tuple of
distinct vertices
$\mathbf{k} = (k_1, \ldots, k_m)$
in
$\ca{G}$,
under the mapping
$u \mapsto k_{\sigma u}$;
and in that case,
$G^\mathrm{func}_{\mathbf{k},\sigma}$
is the average edge weight of that instance:

\vspace{2mm}
\begin{alignat*}{3}
J^\mathrm{func}_{\mathbf{k},\sigma}
	& \vcentcolon= \prod_{\ca{E}_\ca{M}^0} (J_\mathrm{n})_{k_{\sigma u},k_{\sigma v}}
	&& && \prod_{\ca{E}_\ca{M}^\mathrm{s}} J_{k_{\sigma u},k_{\sigma v}}
	\prod_{\ca{E}_\ca{M}^\mathrm{d}} (J_\mathrm{d})_{k_{\sigma u},k_{\sigma v}}\,, \\
G^\mathrm{func}_{\mathbf{k},\sigma}
	& \vcentcolon= \sum_{\ca{E}_\ca{M}^\mathrm{s}} G_{k_{\sigma u},k_{\sigma v}}
	&& + && \sum_{\ca{E}_\ca{M}^\mathrm{d}} (G_\mathrm{d})_{k_{\sigma u},k_{\sigma v}}\,, \\
J^\mathrm{struc}_{\mathbf{k},\sigma}
	& \vcentcolon= \prod_{\ca{E}_\ca{M}^0} (J_0)_{k_{\sigma u},k_{\sigma v}}
	&& && \prod_{\ca{E}_\ca{M}^\mathrm{s}} (J_\mathrm{s})_{k_{\sigma u},k_{\sigma v}}
	\prod_{\ca{E}_\ca{M}^\mathrm{d}} (J_\mathrm{d})_{k_{\sigma u},k_{\sigma v}}\,, \\
G^\mathrm{struc}_{\mathbf{k},\sigma}
	&\vcentcolon= \sum_{\ca{E}_\ca{M}^\mathrm{s}} (G_\mathrm{s})_{k_{\sigma u},k_{\sigma v}}
	&& + && \sum_{\ca{E}_\ca{M}^\mathrm{d}} (G_\mathrm{d})_{k_{\sigma u},k_{\sigma v}}\,,
\end{alignat*}
\ja{\vspace{1mm}}{}

\noindent  and the summations and products are over the missing edges,
single edges and double edges of $\ca{M}$ as follows:

\begin{align*}
	\ca{E}_\ca{M}^0 &\vcentcolon= \{ (u,v) : 1 \leq u < v \leq m : (u,v) \notin \ca{E_M}, (v,u) \notin \ca{E_M} \}\,, \\
	\ca{E}_\ca{M}^\mathrm{s} &\vcentcolon= \{ (u,v) : 1 \leq u < v \leq m : (u,v) \in \ca{E_M}, (v,u) \notin \ca{E_M} \}\,, \\
	\ca{E}_\ca{M}^\mathrm{d} &\vcentcolon= \{ (u,v) : 1 \leq u < v \leq m : (u,v) \in \ca{E_M}, (v,u) \in \ca{E_M} \}\,.
\end{align*}

\end{proposition}
\begin{proof}
See
\Cref{proof:motif_adj_matrix_formula}.
\end{proof}

\subsection{Motif adjacency matrices for bipartite graphs} \label{sec:bipartite}
We also extend our formulation to weighted bipartite networks,
by considering certain 3-node anchored motifs.
These are used to create separate similarity matrices
for each part of a bipartite graph.

\begin{definition}
A \emph{bipartite graph} is a directed graph
where the vertices can be partitioned as $\ca{V} = \ca{S} \sqcup \ca{D}$,
such that
every edge starts in $\ca{S}$ and ends in $\ca{D}$.
We refer to $\ca{S}$ as the \emph{source vertices} and to $\ca{D}$ as the \emph{destination vertices}.
\end{definition}

\noindent Our method for clustering bipartite graphs uses two \emph{anchored} motifs;
the \emph{collider} and the \emph{expander} (\Cref{fig:expa_coll}).
For both motifs the anchor set is $\ca{A}=\{ 1,3 \}$.
These motifs are useful for bipartite clustering because
when restricted to the source or destination vertices,
their
MAMs are the adjacency matrices of the
weighted projections
\citep{chessa2014cluster}
of the graph $\ca{G}$
(\Cref{prop:coll_expa_formulae}).
In particular they can be used as similarity matrices for the source and destination vertices respectively.
Note that if a bipartite graph is connected,
then so are the projections onto its source or destination vertices.

\wf{\begin{figure}[t]
	\centering
	\includegraphics[width=0.3\textwidth]{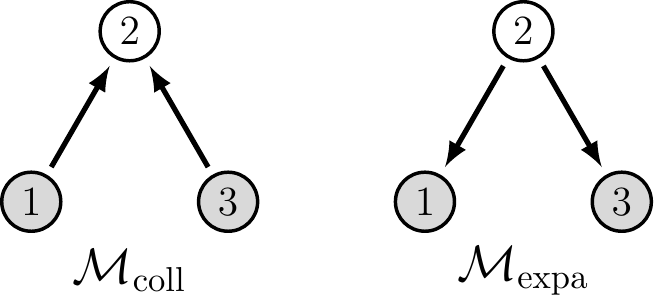}
	\caption{The collider and expander motifs.}
	\label{fig:expa_coll}
\end{figure}
}{}

\begin{proposition}[Colliders and expanders in bipartite graphs] \label{prop:coll_expa_formulae}
Let $\ca{G} = (\ca{V,E},W)$ be a directed bipartite graph. Let $M_\mathrm{coll}$ and $M_\mathrm{expa}$ be the structural or functional MAMs of $\ca{M}_\mathrm{coll}$ and $\ca{M}_\mathrm{expa}$ respectively in $\ca{G}$. Then

\vspace{2mm}
	\begin{align}
		(M_\mathrm{coll})_{ij} &= \bb{I} \{i \neq j\} \hspace*{-0.4cm} \sum_{\substack{k \in \ca{D} \\ (i,k), (j,k) \in \ca{E}}} \hspace*{-0.2cm} \frac{1}{2} \Big[ W((i,k)) + W((j,k)) \Big]\,, \label{eq:mamcoll}\\
		(M_\mathrm{expa})_{ij} &= \bb{I} \{i \neq j\} \hspace*{-0.4cm} \sum_{\substack{k \in \ca{S} \\ (k,i), (k,j) \in \ca{E}}} \hspace*{-0.2cm}\frac{1}{2} \Big[ W((k,i)) + W((k,j)) \Big]\,. \label{eq:mamexpa}
	\end{align}
\end{proposition}
\begin{proof}
See \Cref{proof:coll_expa_formulae}.
\end{proof}\noindent We note that there are other options available for constructing projections of weighted bipartite graphs,
such as the approach in \cite{stram2017weighted},
which, similarly to our framework,
uses sums of edge weights over
shared neighbors.
Another relevant line of work is that of~\citep{zha2001bipartite}, who proposed a certain minimization problem on the bipartite graph, showing that an approximation solution could be obtained via a partial singular value decomposition of a suitably scaled edge weight matrix.

\subsection{Clustering the motif adjacency matrix} \label{sec:clustering_mams}
A motif adjacency matrix can be construed as a general pairwise similarity measure, and as such, can be analyzed directly as a new weighted \emph{undirected} graph. Following~\citep{benson2016higher},
one of the most interesting applications  is identifying higher-order clusters with spectral methods,
leveraging the fact that the similarity is based on motifs.

To extract clusters, we use a standard approach from the literature on spectral clustering,
based on the spectrum of the random-walk Laplacian, followed by $k$-means++  (see \Cref{chap:spectral}).
To this end, we need to define two parameters: $l$ is the number of random-walk Laplacian
eigenvectors to use, and
$k$ is the number of clusters for
$k$-means++
(see \Cref{alg:rwspectclust}).
We detail our approach for general directed weighted networks in \Cref{alg:motifrwspectclust},
and for weighted bipartite networks in \Cref{alg:bipartite_clustering}.

\subsubsection{Cluster evaluation}\label{sec:clustereval}
When ground-truth clustering is available, we compare it to our recovered clustering using the  \emph{Adjusted Rand Index} (ARI) \citep{hubert1985comparing}.
The ARI between two clusterings has expected value $0$ under random cluster assignment,
and maximum value $1$ denoting perfect agreement between the clusterings.
A larger ARI indicates a more similar clustering, and hence closer to the ground truth.

\subsubsection{Connected components}\label{sec:connected_components}
In order for spectral clustering to produce nontrivial clusters from a graph,
it is necessary to restrict the graph to its largest connected component.
When forming MAMs,
even if the original graph is
\green{(weakly or strongly)}
connected,
there is no guarantee that the
\green{(symmetric)}
MAM is connected too.
Hence we restrict the MAM to its largest connected
component before spectral clustering is applied.
While this may initially seem to be a flaw with motif-based spectral clustering (since some vertices may not be assigned to any cluster),
in fact it can be useful:
we only attempt to cluster vertices which are in some sense ``well connected'' to the rest of the graph.
This can result in fewer misclassifications
than with traditional spectral clustering,
as seen in Section~\ref{sec:motif_polblogs},
\green{where we also investigate
MAM regularization as an alternative strategy for dealing with disconnected MAMs.}

\subsubsection{Motif choice}
\label{sec:motif_choice}

\green{
The selection of a motif is essentially equivalent to selecting a similarity measure between
nodes in the network,
as the $(i,j)$th entry in the MAM corresponds to the similarity between nodes $i$ and $j$ used for our clustering procedure.
This selection can be important, as we will see in
\Cref{sec:syntheticexample2},
where the choice of motif can have a significant impact on the clusters obtained.
For motif selection, considerations for weighted networks are largely
similar to those for unweighted networks.
Thus much of this discussion will mirror that of
\cite{benson2016higher},
although we will highlight the areas where weight may cause deviations.
}

\green{
The first criterion for motif selection is the application domain.
This may involve choosing motifs which are related to specific features of the application domain,
essentially specializing the similarity measure to the task at hand.
For example,
in several fields, some motifs may be more relevant or have very different properties;
e.g. the feed forward structure in biological networks
\citep{Mangan11980},
various triadic structures in sociology
\citep{wasserman1994social},
or the motif $\ca{M}_5$ in food web networks
\citep{benson2016higher}.
For weighted networks
the procedure is very similar,
but care must be taken to simultaneously select a motif
and a weighting scheme
(\Cref{sec:weighted_mams}),
in order to capture
the similarity measure of interest.
}

\green{
Again following \cite{benson2016higher},
in the case where there is no application domain guidance,
more principled methods for motif choice are also available.
For example,
sweep profiles \citep{shi2000normalized}
or their motif-coherence counterparts \citep{benson2016higher}
can be used to
identify motifs which
give more clearly distinguished clusters.
In the weighted case,
these methods require
weighted generalizations of the appropriate quantities.
We give an example of the
application of
weighted generalizations of
these sweep profiles
to real world data in
\Cref{sec:motif_polblogs}.
Finally,
one could consider every motif (or a subset of motifs) and then
explore each in turn.
We demonstrate this approach in our migration
experiment in \Cref{sec:motif_migration},
where we plot the geographic spread of our clusters and then validate on a subset of motifs by considering the cut imbalance ratio scores.}

\subsubsection{Functional vs. structural MAMs}
There is also a choice of whether to use functional or structural MAMs
for motif-based clustering,
and their different properties make them suitable for different circumstances.
Firstly, note that $ 0 \leq M^\mathrm{struc}_{ij} \leq M^\mathrm{func}_{ij}$
for all $i,j \in \ca{V}$.
This implies that the largest connected component of $M^\mathrm{func}$
is always at least as large as that of $M^\mathrm{struc}$, meaning that
sometimes
more vertices can be assigned to a cluster by using functional MAMs. However, structural MAMs are more discerning about finding motifs, since they require both ``existence'' and ``non-existence'' of edges.

\subsection{Computational analysis}\label{sec:computational_analysis}

For motifs on at most three vertices,
we provide two fast and
potentially parallelizable
matrix-based
procedures for computing weighted MAMs:
one for dense regimes
and one for sparse regimes.
In this section, we explore and demonstrate the scalability of these two approaches.

\subsubsection{Dense approach}
First, \Cref{prop:motif_adj_matrix_computation}
bounds the number of matrix operations required
to compute a weighted MAM,
for a motif on at most three vertices,
using our dense formulation.
In practice, additional symmetries of the motif often allow
computation with even fewer matrix operations
(\Cref{chap:appendix_matrices}).

\begin{proposition}[Complexity of MAM formula] \label{prop:motif_adj_matrix_computation}Let $\ca{G}$ be a (weighted, directed) graph
on $n$ vertices,
and suppose that
the $n \times n$ directed adjacency matrix
$G$ of $\ca{G}$ is known.
Then, computing the adjacency and indicator matrices and calculating a MAM,
for a motif on at most three vertices, using Equations
(\ref{eq:mamformulafunc})
and
(\ref{eq:mamformulastruc})
in \Cref{prop:motif_adj_matrix_formula},
involves at most 18 matrix multiplications,
22 entry-wise multiplications and 21 additions of
$n \times n$ matrices.
\end{proposition}

\begin{proof}
See \Cref{proof:motif_adj_matrix_computation}.
\end{proof}

\noindent Thus as we have a fixed bound on the number of each operation for any motif,
and as each operation is performed on a matrix of the same size,
our approach scales with the largest complexity of these operations.
In dense matrices, element-wise products and additions are $\ca{O}(n^2)$
and naive matrix multiplication is $\ca{O}(n^{3})$. Therefore, in a naive implementation
for a graph on $n$ vertices, the overall complexity of our dense approach is $\ca{O}(n^{3})$,
and the memory requirement is $\ca{O}(n^2)$.
We note that somewhat faster algorithms are available for multiplication of large dense matrices,
such as the $\ca{O}(n^{2.81})$ algorithm given by
\cite{strassen1969gaussian}.

A list of functional MAM formulae for our dense approach, for all simple motifs on at most three vertices,
as well as for the collider and expander motifs (used in \Cref{sec:bipartite}), is given in \Cref{tab:motif_adj_mat_table}
in \Cref{chap:appendix_matrices}.
These formulae are generalizations of those stated in Table S6 in the supplementary materials of \cite{benson2016higher} (a list of structural MAMs for triangular motifs in unweighted graphs). Note that the functional MAM formula for the two-vertex motif $\ca{M}_\mathrm{s}$ yields the symmetrized adjacency matrix $M = G + G^\mathsf{T}$, which can be used for traditional spectral clustering
(\Cref{sec:spectral_overview}), as in
\cite{Meila2007ClusteringBW}.

\subsubsection{Sparse approach}
Real world networks are often sparse
\citep{snapnets},
and operations with sparse matrices are
significantly faster:
in sparse matrices with $b$
non-zero entries, element-wise products and additions are $\ca{O}(b)$,
and matrix multiplications are $\ca{O}(bn)$.
Therefore, we give a slightly different approach with a better computational running time
(\Cref{sec:empirical_computation})
on sparse graphs,
which can extend to graph sizes that are infeasible for our dense approach.

For motifs on at most
three vertices,
the formulae in \Cref{prop:motif_adj_matrix_formula}
are simply sums of terms of the form
$A \circ (BC)$,
where $A,B$ and $C$ are adjacency or indicator matrices, and $\circ$ represents the entry-wise (Hadamard)
product.
If the directed graph has $n$ vertices
and $b$ edges,
then most of the adjacency and indicator matrices
have at most $b$ non-zero entries.
For these matrices, we can compute the
(possibly dense) matrix $BC$
in $\ca{O}(bn)$ time,
and then the matrix
$A \circ (BC)$
in $\ca{O}(n^2)$ time.
Summing them together is also done
in $\ca{O}(n^2)$ time.\footnote{By exploiting sparsity
the sums and element-wise products
could be faster than $\ca{O}(n^2)$,
but as the complexity is dominated by the matrix multiplication,
we do not explore this further.}

However, when considering motifs with missing edges,
we use the \emph{dense} indicator matrices
$J_0$ and $J_\mathrm{n}$. This is problematic as now the term
$BC$ can contain a dense matrix, which can be an issue both for computational reasons\footnote{Although there are sparse-dense matrix multiplication algorithms that address this.} and more importantly, for memory requirements.
To address this, we rewrite these two matrices as

\begin{align*}
J_\mathrm{n} & = {\bf 1} - I\,,
\\
J_0 & = {\bf 1} - (I + J_\mathrm{s} + J_\mathrm{s}^\mathsf{T} + J_\mathrm{d})
\\
& = {\bf 1} - (I + \tilde{J})\,,
\end{align*}

\noindent  where ${\bf 1}$ is a matrix of 1s,
and expand out the resulting formulae.
Element-wise and matrix products with ${\bf 1}$ and $I$
are at most $\ca{O}(n^2)$ so
are computationally simple,
and do not require generating dense matrices.
This allows scaling simply using standard linear algebra libraries.

For sparse graphs the key limitation of this approach
is in fact not CPU, but memory (RAM) as
while $B$ and $C$ might be sparse,
$BC$ may not be.
The worst case density of $BC$ for $B$ and $C$ indicator matrices of a connected graph
is $n^2$ (e.g. a star graph).\footnote{\scriptsize See \url{https://math.stackexchange.com/questions/1042096/bounds-of-sparse-matrix-multiplication}.}
However, in practice the density is often substantially lower than this,
allowing our approaches to scale to very large graphs
(\Cref{sec:empirical_computation}).

\subsubsection{Empirical computational speed}
\label{sec:empirical_computation}

\wf{\begin{figure}[t]
  \centering
  \includegraphics[width=0.99\textwidth]{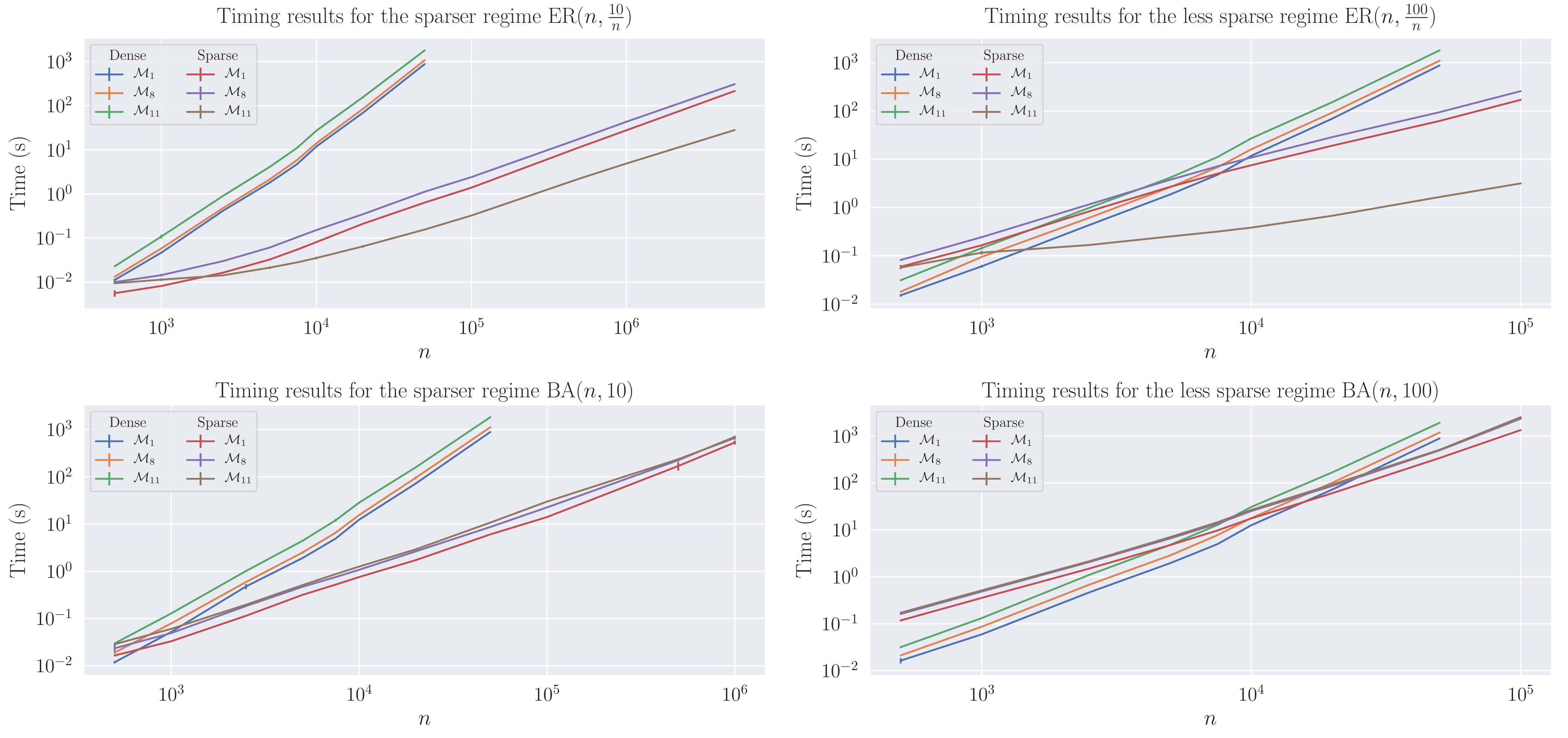}
  \caption{\green{Run time for our two approaches.
  {\bf Top:} Erd\H{o}s-R\'enyi (ER) graph ensembles.
  {\bf Bottom:} Barab\'asi-Albert (BA) graph ensembles.
  The left panels show very sparse graphs
  (ER: $p=\frac{10}{n}$, BA: $m=10$)
  and the right panels show less sparse graphs
  (ER: $p=\frac{100}{n}$, BA: $m=100$).
  We average over five repeats, and error bars are one sample standard deviation.}}
  \label{fig:runTimes}
\end{figure}
}{}

\noindent In this section,
we demonstrate the scalability of our two approaches to computing MAMs.
We showcase the running times of
\green{four different random graph ensembles.
For the first two ensembles, we use the directed Erd\H{o}s-R\'enyi (ER) model \citep{erdos1959random} in which each directed edge between $n$ nodes exists independently with probability $p$. We consider two regimes,
a very sparse regime
($p=\frac{10}{n}$)
and a less sparse regime
($p=\frac{100}{n}$).
For the last two  ensembles,
we use the
undirected Barab\'asi-Albert (BA)
model~\citep{barabasi1999emergence}.
In this model, the graph is initialized with $m$ nodes,
and $n-m$ nodes are then added sequentially,
with each node connecting to $m$ previously placed nodes, and with probability
of connection proportional to the current degree of each previously placed node,
resulting in a graph with a skewed (power-law) degree distribution.
We again consider two regimes,
a very sparse regime
($m = 10$)
and a less sparse regime
($m = 100$).
In each of these four ensembles,
the expected number of edges in the graph scales as $\ca{O}(n)$.}

We measure the performance of computing a functional MAM for a representative sample of motifs:
one triangle motif $\mathcal{M}_1$,
the 2-star $\mathcal{M}_{8}$,
and a motif with a
\green{bi-directional}
edge, $\mathcal{M}_{11}$.
We do not include the time taken to generate the graph in either the sparse or the dense matrix form.
To make this a fair test, we use an optimized Python environment
from the Data Science Virtual Machine image on an
\green{
E64s v3
Azure virtual machine with 64 vCPUs
and 432\,GiB of RAM.}

The results are summarized in \Cref{fig:runTimes}.
\green{In the denser graphs
(ER: $p=\frac{100}{n}$, BA: $m=100$),
we note that for smaller graphs
(i.e for $n$ less than $\approx 10^{4}$) the dense approach
tends to perform as well as, and in many cases exceeds, the performance of the sparse approach, highlighting the advantage of using this approach in denser graphs.}
\green{For graphs above this size and graphs in the sparser regime (ER: $p=\frac{10}{n}$, BA: $m=10$),
the sparse approach outperforms,}
handling graphs with
\green{over a million nodes} and tens of millions of edges.
\green{
In the less sparse regime (ER: $p=\frac{100}{n}$, BA: $m=100$), we consider graphs of size up to $n=10^5$ nodes and of order $10^7$ edges, for which our sparse approach takes less than $10^2$ seconds for the
ER model,  and around $10^3$ seconds for the BA model.}

We note that the time required for each method is roughly constant
across most of the tested motifs. The exception is the sparse approach for $\ca{M}_{11}$,
which takes substantially less time than the other motifs in both
\green{of the ER
regimes. We believe this to be related to the
bi-directional edge in $\ca{M}_{11}$
(not present in $\ca{M}_1$ or $\ca{M}_8$),
which is rare in sparse ER graphs,
heavily reducing the amount of computation required.
This is not seen in the undirected BA graphs,
where bi-directional edges are much more common, resulting in a similar level of
performance across all tested motifs in this regime.}
The various scenarios we experimented with
highlight the fact that our approach is highly scalable to very large sparse graphs,
\green{including those with skewed degree sequences
which often arise in real world applications.}

 \section{Applications to Synthetic Data}
\label{sec:motif_dsbms}

To validate our method, we consider three synthetic data sets as examples. We selected each example to highlight a
different aspect of our approach and to demonstrate the advantages of considering a weighted higher-order clustering measure.
Example 1 demonstrates the importance of taking the edge weights into account when performing higher-order clustering;
Example 2 shows the value of higher-order clustering, demonstrating that clustering using different motifs can yield different insights into the same data; and finally, Example 3 demonstrates the value of our bipartite clustering scheme for detecting structures in weighted bipartite networks.

\subsection{Directed stochastic block models}

For these tests we use weighted and unweighted \emph{directed stochastic block models} (DSBMs),
a broad class of generative models for directed graphs \citep{nowicki2001estimation}.
An unweighted DSBM is characterized by a block count $k$,
a list of block sizes $(n_i)_{i=1}^k$,
and a connection matrix $F \in [0,1]^{k \times k}$.
We define the cumulative block sizes $n^\mathrm{c}_i = \sum_{j=1}^i n_j$,
and the total graph size $n=n^\mathrm{c}_k$.
These are used to construct the group allocations
$g_i = \min \{r: n^\mathrm{c}_r \geq i\}$,
and finally a graph $\ca{G}$ is generated with adjacency matrix entries

\begin{equation}
G_{ij} = \textrm{Ber}(F_{g_i,g_j}) \cdot \bb{I}\{i \neq j\},
\label{eq:unweighted}
\end{equation}

\noindent with all Bernoulli random variables sampled independently.
A weighted DSBM is constructed in a similar manner
(see for example
\cite{mariadassou2010uncovering}
and
\cite{clausetWeightedBlockModel}),
but also requires a  weight matrix $\Lambda \in [0,\infty)^{k \times k}$. In this case, the weighted adjacency matrix entries are generated by the following mixture model

\begin{equation}
G_{ij} = \textrm{Ber}(F_{g_i,g_j})
\cdot \textrm{Poi}(\Lambda_{g_i,g_j})
\cdot \bb{I}\{i \neq j\},
\label{eq:weighted}
\end{equation}

\noindent  with all Bernoulli and Poisson random variables sampled independently.
Note that the Bernoulli variable now no longer directly corresponds to
edge existence, since there is always a non-zero chance of the Poisson
variable being zero, which sets $G_{ij}=0$ and thus removes the edge.
We assume a DSBM is unweighted unless stated otherwise. We evaluate performance using the Adjusted Rand Index (ARI) \citep{rand1971objective}, as in \Cref{sec:clustereval}.
\subsection{Bipartite stochastic block models}
\label{sec:bipartite_sbm}

We define the unweighted \emph{bipartite stochastic block model} (BSBM)
with source block count $k_\ca{S}$,
destination block count $k_\ca{D}$,
source block sizes $(n_\ca{S}^i)_{i=1}^{k_\ca{S}}$,
destination block sizes $(n_\ca{D}^i)_{i=1}^{k_\ca{D}}$,
and bipartite connection matrix $F_\textrm{b} \in [0,1]^{k_\ca{S} \times k_\ca{D}}$
as the unweighted DSBM with block count
$k = k_\ca{S} + k_\ca{D}$,
block sizes
$(n_i)_{i=1}^{k} =
\big((n_\ca{S}^i)_{i=1}^{k_\ca{S}},
(n_\ca{D}^i)_{i=1}^{k_\ca{D}}\big)$,
and connection matrix
$F =
\begin{psmallmatrix}
0 & F_\mathrm{b} \\
0 & 0
\end{psmallmatrix}
$.

A weighted BSBM with
bipartite weight matrix $\Lambda_\textrm{b} \in [0,\infty)^{k_\ca{S} \times k_\ca{D}}$
can similarly be constructed by using a weighted DSBM with
weight matrix
$\Lambda =
\begin{psmallmatrix}
0 & \Lambda_\mathrm{b} \\
0 & 0
\end{psmallmatrix}
$.
Note that this is a generalization of the model in \cite{florescu2016spectral}.

\subsection{Example 1}

In this example we demonstrate the advantages of taking edge weights into account when detecting higher-order structures.
We compare \Cref{alg:motifrwspectclust} with two clusters ($k=2$) and two eigenvectors ($l=2$), against two standard  approaches. The first comparison is random-walk spectral clustering using a symmetrized weighted matrix, which captures the ability of non-motif-based methods to uncover the underlying structure. This is equivalent to \Cref{alg:motifrwspectclust} with motif $\ca{M}_\mathrm{s}$. Secondly, we compare against our own motif-based approaches, but with all edge weights set to $1$, similar to the formulation of \cite{benson2016higher}.

\wf{\begin{figure}[t]
    \centering
    \includegraphics[width=0.87\textwidth]{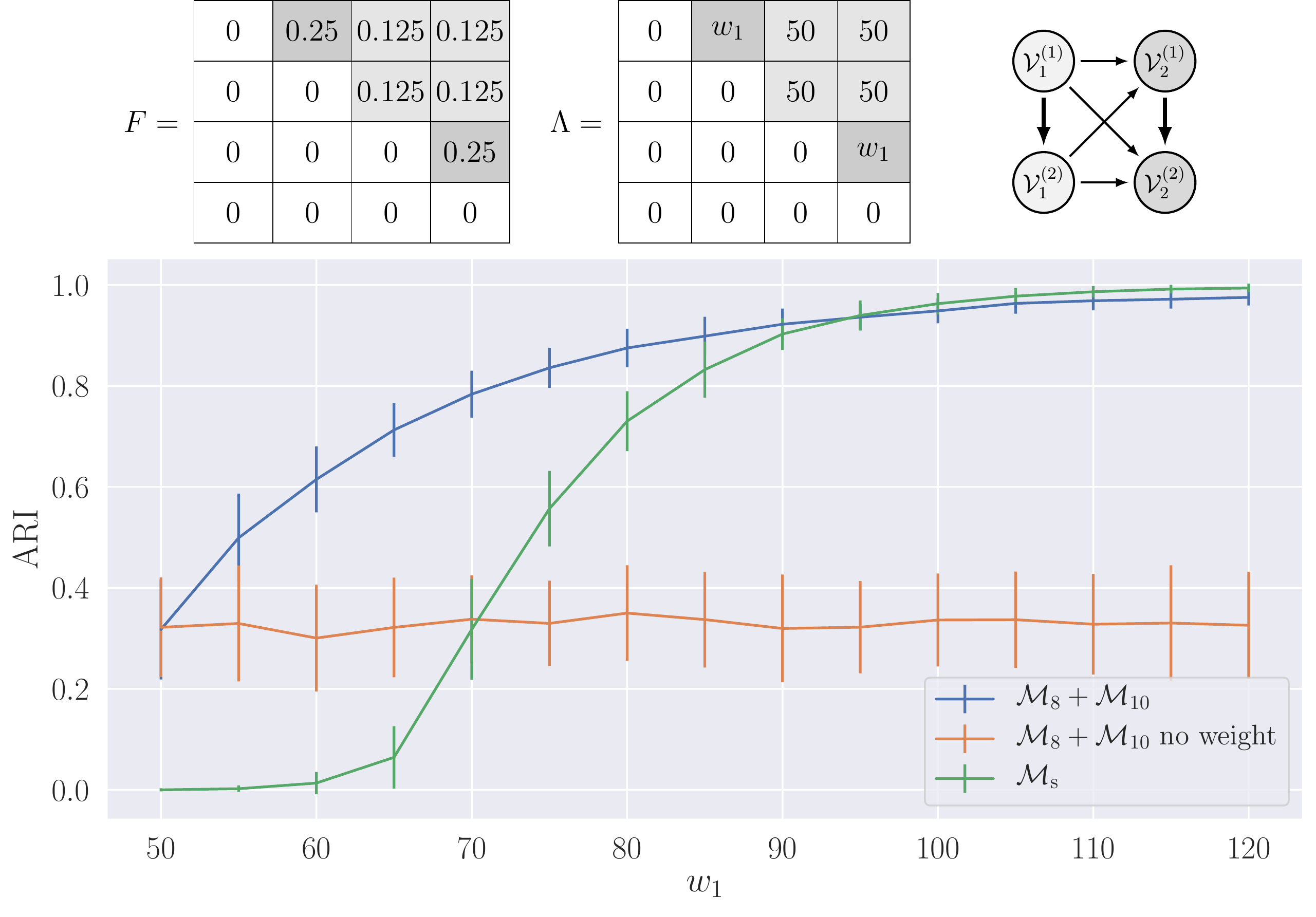}
	\caption{{\bf Top:} DSBM parameters for Example 1. We note that the DSBM has been constructed to favor weighted motifs.
	The left and middle panels display the connection matrix and  weight matrix (\Cref{eq:weighted}).
	In the right panel, a schematic diagram of the structure, blocks of the same color belong to the same cluster and larger arrows represent connections with higher probabilities and larger edge weights.
	{\bf Bottom:} Exploring the performance of MAMs based on $\mathcal{M}_8$ and $\mathcal{M}_{10}$ on the model in the upper panel. Each block contains $100$ nodes. We compare both to the unweighted case, and to the symmetrized case. We perform $100$ repeats, and error bars are one sample standard deviation.
	}
	\label{fig:weighted}
\end{figure}
}{}

To demonstrate the advantages of our approach, we consider an example for which
(i) a higher-order structure is present, and
(ii) weight is important in the structure.
Constructing higher-order structures in a stochastic block model is challenging, as by definition all the edges are independent. Thus the block structure in a DSBM with strongly connected blocks is
captured by density rather than by the existence of motifs (although these are correlated).

To this end, we construct clusters that consist of several blocks,  and introduce a higher-order structure between blocks of the same cluster. For this example, we use two clusters (\Cref{fig:weighted} (upper panel)), each one consisting
of two blocks each of size $100$, for a total of $n=400$ nodes.
Each cluster ($\{\mathcal{V}^{(1)}_1,\mathcal{V}^{(2)}_1\}$
and $\{\mathcal{V}^{(1)}_2,\mathcal{V}^{(2)}_2 \} $)
consists of two blocks with strong uni-directional
(probability $p=0.25$) connections with
potentially large weights (Poisson with mean $w_1$).
The two clusters are then linked by weaker inter-block connections
with potentially smaller weights (Poisson with mean $50$).
By design, this model has strong heavy-weighted uni-directional structures, and is thus well captured by $\ca{M}_8$ and $\ca{M}_{10}$, both of which capture uni-directional structure.
Thus,  following~\citep{benson2016higher}
rather than focusing on an individual motif,
we use the sum of both MAMs.

The lower panel of \Cref{fig:weighted} displays the results for functional motifs (structural motifs in	\Cref{sec:additionSyn}). As our procedure only clusters the largest connected MAM component, we compute ARI over this component.
For $w_1=50$,
all edges have the same expected weight,
and thus the performances of weighted and non-weighted motifs are equal.
In the mid-range of weights ($50 < w_1 \leq 90$), our approach outperforms both of the others, indicating the advantages of accounting for weights.

Finally, for large weights ($w_1>90$), we are (on average) slightly outperformed by $\ca{M}_\mathrm{s}$,
the symmetrized weighted adjacency matrix,
although the error bars are overlapping.
One possible reason for this is that
while using $\ca{M}_8$ and $\ca{M}_{10}$ gives a strong signal,
it also introduces noise when
motifs with heavy and non-heavy weighted edges span clusters.

The pattern on structural motifs is similar
(\Cref{fig:weightedStruc} in \Cref{sec:additionSyn}),
with two key differences:
first, the non-weighted motifs have a larger ARI
($\approx 0.7$ vs. $\approx 0.3$),
and second, the weighted motif equals or outperforms $\mathcal{M}_s$.
\green{We hypothesize the following
possible
reason for this phenomenon:
in this model,
motifs on three nodes which span clusters can form triangles,
whereas those within clusters cannot
(due to the probability-$0$ connections).
Hence
versions of the non-triangle motifs $\mathcal{M}_8$ and $\mathcal{M}_{10}$
are able to filter out some of the noise
described in the previous paragraph,
leading to better performance.}

Finally, we note that this example has been designed to highlight the advantages of our motifs;
and several other structures were considered which do not have this property.
When applying this to real world data, it is important to consider the correct motif to use,
as discussed in \Cref{sec:motif_choice}.

\subsection{Example 2}\label{sec:syntheticexample2}
In the second example, we explore the value of higher-order clustering, by demonstrating that clustering using  different motifs can give different albeit potentially equally valuable insights.
We showcase this with two experiments. First, we consider the case where, due to the construction of the network, certain motifs are better at detecting the underlying structure in different parameterization regimes.
Second, expanding on our first example, we consider a DSBM where, by construction, the network does not have strong densely connected blocks, and thus different motifs highlight distinct and  equally relevant groupings in the network.

\subsubsection{Experiment 1}

For the first experiment, we use the DSBM with $k=2$, $n_1=n_2=100$
(so $n=200$)
and $F =
\begin{psmallmatrix} 0.2 & q_1 \\ 0.05 & 0.2 \end{psmallmatrix}$, as depicted
in the upper panel of \Cref{fig:benchmark1}.

\wf{\begin{figure}[t]
  \centering
  \includegraphics[width=0.83\textwidth]{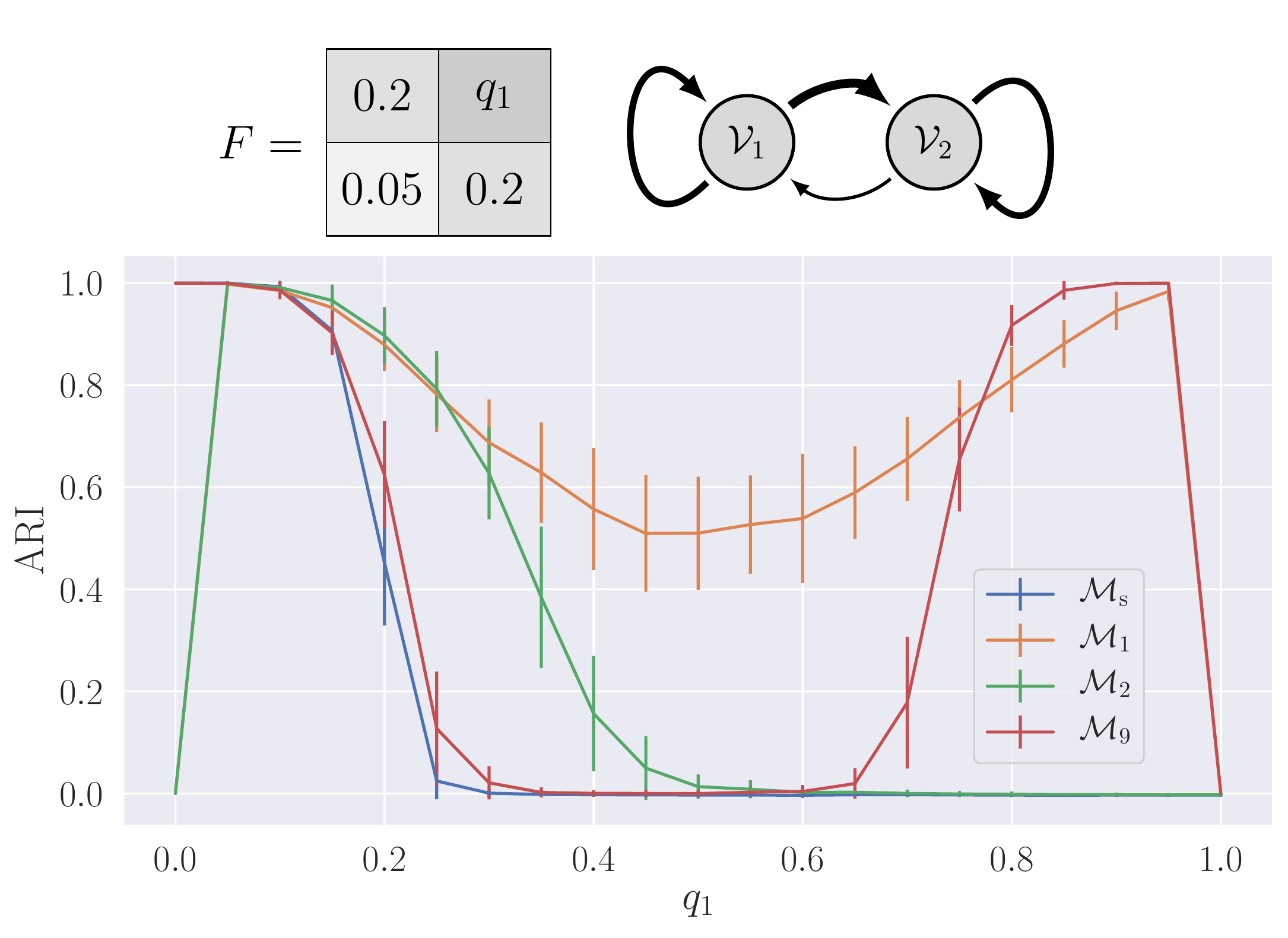}
  \caption{{\bf Top:} DSBM parameters for Example 2 Experiment 1.
    Left panel, the connection matrix for this model (\Cref{eq:unweighted}).
    Right panel, a schematic diagram of the structure: larger arrows represent connections with higher probabilities.
    {\bf Bottom:}
    Performance on this benchmark using \green{structural} MAMs. We compare with standard symmetrization $\mathcal{M}_s$.
    We perform
    $100$
    repeats, and error bars are one sample standard deviation.
  }
  \label{fig:benchmark1}
\end{figure}
}{}

We test the performance of \Cref{alg:motifrwspectclust}
across a few structural motifs with parameters $k=l=2$ on this model
(functional motifs in
\Cref{sec:additionSyn}).
It can be seen that the motifs perform well in different regions,
with $\mathcal{M}_1$ outperforming other methods in the regime $0.3 \le q_1 \le 0.7$,
and also performing well outside this range.
For large values of $q_1$, motif $\mathcal{M}_9$ performs best.
For lower values of $q_1$, $\mathcal{M}_2$ has a slightly higher average ARI,
although the difference is within the standard errors.
This altogether demonstrates the importance of selecting the right motif. Furthermore, we compare against $\mathcal{M}_s$, the symmetrized weighted adjacency matrix, and
each presented motif outperforms this baseline.

We also observe an artifact in this plot: for certain motifs, the performance drops away at
$q_1=0$ (triangles) and at $q_1=1$ (all motifs).
For these parameter values,
the MAM becomes entirely disconnected, with each of the two groups in its own connected component. As our clustering  scheme only considers the largest connected component (\Cref{sec:connected_components}),
we then obtain an ARI of $0$. In real world graphs, we would recommend investigating all reasonably large connected components
\green{or, depending on the application, employing some form of regularization
(see \Cref{sec:motif_polblogs}).}

We note that there is a bi-modality with certain motifs,
notable with $\mathcal{M}_9$ performing well for small and large values of $q_1$, a feature not present within the functional motifs (\Cref{fig:benchmark1FUNC}).
We believe this is due to the fact that structural motifs act as a filter: with high values of $q_1$, it is difficult to form 2-paths ($\mathcal{M}_9$) between groups without also forming
triangular motifs, which would not contribute towards a structural MAM.

\subsubsection{Experiment 2}

In this experiment, we demonstrate the ability of MAMs to uncover different structures.
We construct a DSBM with $8$ groups
of size $100$ ($n=800$)
with block structure given by the diagram
in the left panel of \Cref{fig:benchmark1Fig}:
we place edges which match the block structure with probability $0.5$
(i.e. $F_{ab}=0.5$ if there is an arrow from $a$ to $b$ in the
diagram),
and $0.0001$ otherwise, in order  to maintain connectivity.

\wf{\begin{figure}[t]
  \centering
  \includegraphics[width=0.85\textwidth]{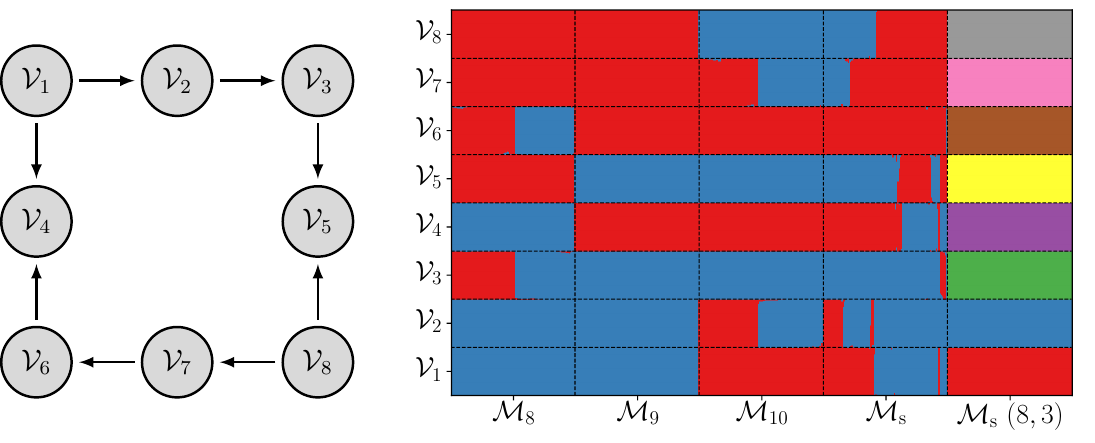}
  \caption{{\bf Left}:
    Diagram of the block structure for Example 2 Experiment 2: edges in the diagram are present with probability $0.5$, all other edges are present with probability $0.0001$ (including edges within blocks).
    {\bf Right}: The detected groups found by each method, with $k=2$ and $l=2$, with the exception of the last column which has $k=8$, $l=3$.
    We test $100$ replicates and present the results as columns in the plot.
    The nodes are ordered by block, and colors represent the group
    allocations in a given block for each motif-based method.
    We order replicates (i.e. columns) within each motif to highlight similarities,
    the columns are sorted within block in order to
    promote
    contiguity of the colors across each block.
    While the clustering assignments may contain some errors,
    the results are relatively robust across replicates.
  }
  \label{fig:benchmark1Fig}
\end{figure}
}{}

Following Example 1, in this DSBM we do not have the usual strongly connected groups:
in fact, each group has a close to zero probability
($0.0001$) of within-group connections.
Thus, while there are clearly $8$ groups with unique connection patterns,
there does not exist a ``correct'' division of the nodes into densely connected groups, but rather, different motifs highlight different structures.

We display the results of using functional motifs with \Cref{alg:motifrwspectclust} with $k=l=2$ in \Cref{fig:benchmark1Fig} (structural results are similar in
\Cref{fig:benchmark1FigStruc}).
Each column represents the structure uncovered by one of the motifs. For robustness, we perform the experiment on $100$ replicates, and place the resulting partitions side by side in each column. For each motif, the columns are sorted within block in order to
promote
contiguity of the colors across each block. 

We observe that (by design) this graph has $3$ different clusterings,
highlighted by different
motifs. First, $\mathcal{M}_8$ (a 2-star with the edges pointing outwards (\Cref{fig:motif_definitions_directed})).
Each pair of connected blocks is connected by this motif.
However, there is a much stronger connection when this motif is part of the higher-order structure. Thus, when we cluster using $\mathcal{M}_8$, we obtain two consistent clusters consisting of
$\{\mathcal{V}_1, \mathcal{V}_2, \mathcal{V}_4\}$
and $\{\mathcal{V}_5, \mathcal{V}_7, \mathcal{V}_8\}$,
each of which is
united
by this motif at a block level.
The remaining blocks without this block level structure ($\mathcal{V}_3$ and
$\mathcal{V}_6$) are then
essentially randomly
placed into one of the two clusters,
giving the behavior observed in
\Cref{fig:benchmark1Fig}.
We observe a similar structure with the other motifs:
$\mathcal{M}_{9}$ (a 2-path) splits the graph into the two groups
characterized by their 2-paths,
$\{\mathcal{V}_1, \mathcal{V}_2, \mathcal{V}_3,\mathcal{V}_5\}$
and
$\{\mathcal{V}_4, \mathcal{V}_6, \mathcal{V}_7,\mathcal{V}_8\}$;
and finally $\mathcal{M}_{10}$ (a 2-star with the edges pointing inwards)
splits the graph into the two clusters based on this structure,  namely,
$\{\mathcal{V}_1, \mathcal{V}_4, \mathcal{V}_6\}$
and
$\{\mathcal{V}_3, \mathcal{V}_5, \mathcal{V}_8\}$
(with a similar behavior to $\mathcal{M}_{8}$ for the remaining blocks).
Thus, we conclude that our MAMs can obtain very different and equally valid structures in the same graph.

Finally, we compare with $\mathcal{M}_{s}$ which is equivalent to clustering with $G+G^\top$. Under
symmetrization, this graph is a ring of indistinguishable blocks. Therefore, the best possible division is to arbitrarily divide the ring into two roughly equal-sized pieces. Considering \Cref{fig:benchmark1Fig}, we observe this behavior with many different divisions, and no pair of blocks is consistently placed within the same cluster.
For completeness we also compare against $\mathcal{M}_s$ with $k=8$ and $l=3$, which divides the network into the 8 blocks of the underlying DSBM.
Although this is a specially constructed example,
it highlights the different equally valid structures that can be uncovered,
emphasizing the importance of considering the correct motif
\green{(see \Cref{sec:motif_choice})}
and demonstrating the value of this procedure above standard spectral clustering.

\subsection{Example 3}
\label{sec:example3}

In our final example,  we demonstrate clustering the source vertices of bipartite networks, using the collider motif as in
\Cref{alg:bipartite_clustering}.
Clustering the destination vertices with the
expander motif is exactly analogous
(by simply reversing edge directions),
so we do not demonstrate it explicitly.
Note also that in bipartite graphs,
all instances of the collider or expander
motifs are structural, so we need
not compare functional and
structural MAMs.

We illustrate this with two experiments.
In the first one, we show that in certain networks, edge weights are important for our collider method to perform well.
In the second experiment,  we compare against an alternative method for clustering weighted
bipartite networks, showing that under certain conditions, our collider method performs better.

\wf{\begin{figure}[t]
    \centering
    \includegraphics[width=0.85\textwidth]{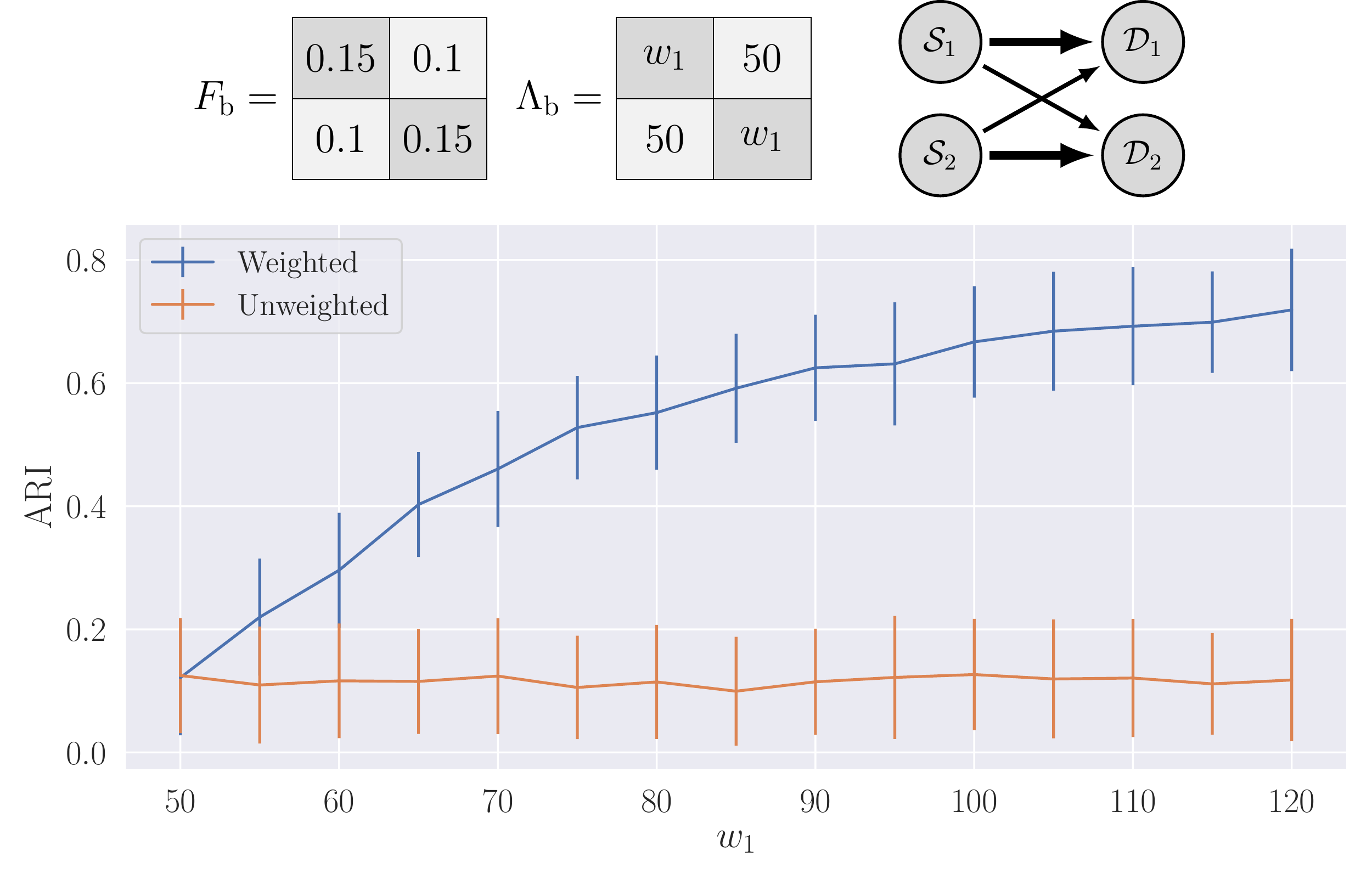}
    \caption{
      \textbf{Top:} BSBM parameters for Example 3 Experiment 1.
      Left panel is the connection matrix, middle panel is the weight matrix (see \Cref{eq:weighted} and \Cref{sec:bipartite_sbm} for details),
      and right panel is a schematic diagram of the structure.
      \textbf{Bottom:} performance of our bipartite collider method with and without edge weights.
      We perform 100 repeats,
      and error bars are one sample standard
      deviation.
    }
    \label{fig:bipartite_weight}
  \end{figure}
}{}

\subsubsection{Experiment 1}

This example illustrates the advantages of using weights for bipartite clustering. We use a weighted BSBM with block counts
$k_\ca{S} = k_\ca{D} = 2$,
block sizes
$n_\ca{S}^1 = n_\ca{S}^2 = n_\ca{D}^1 = n_\ca{D}^2 = 100$,
bipartite connection matrix
$F_\mathrm{b} =
\begin{psmallmatrix}
0.15 & 0.1 \\
0.1 & 0.15 \\
\end{psmallmatrix}$,
and bipartite weight matrix
$\Lambda_\mathrm{b} =
\begin{psmallmatrix}
w_1 & 50 \\
50 & w_1 \\
\end{psmallmatrix}$,
where $w_1$
is a varying parameter.
This network has been constructed so that
vertices in
$\ca{S}_1$
are slightly more likely to connect to
vertices in
$\ca{D}_1$
than to vertices in
$\ca{D}_2$,
and vice versa for
$\ca{S}_2$.
This makes
the source vertex clusters
$\ca{S}_1$
and
$\ca{S}_2$
weakly distinguishable when not considering
edge weights.
However, the edge weights exhibit the same preferences,
and this effect becomes stronger as $w_1$
increases.

We test the performance of \Cref{alg:bipartite_clustering}
for clustering the source vertices (using the collider motif),
with parameters $k_\ca{S} = l_\ca{S} = 2$ on this model.
We compare the results against those obtained by ignoring the edge weights.
\Cref{fig:bipartite_weight} shows that the weights in this model
allow the structure to be recovered well when the weighting effect is large enough.

\wf{\begin{figure}[t]
    \centering
    \includegraphics[width=0.85\textwidth]{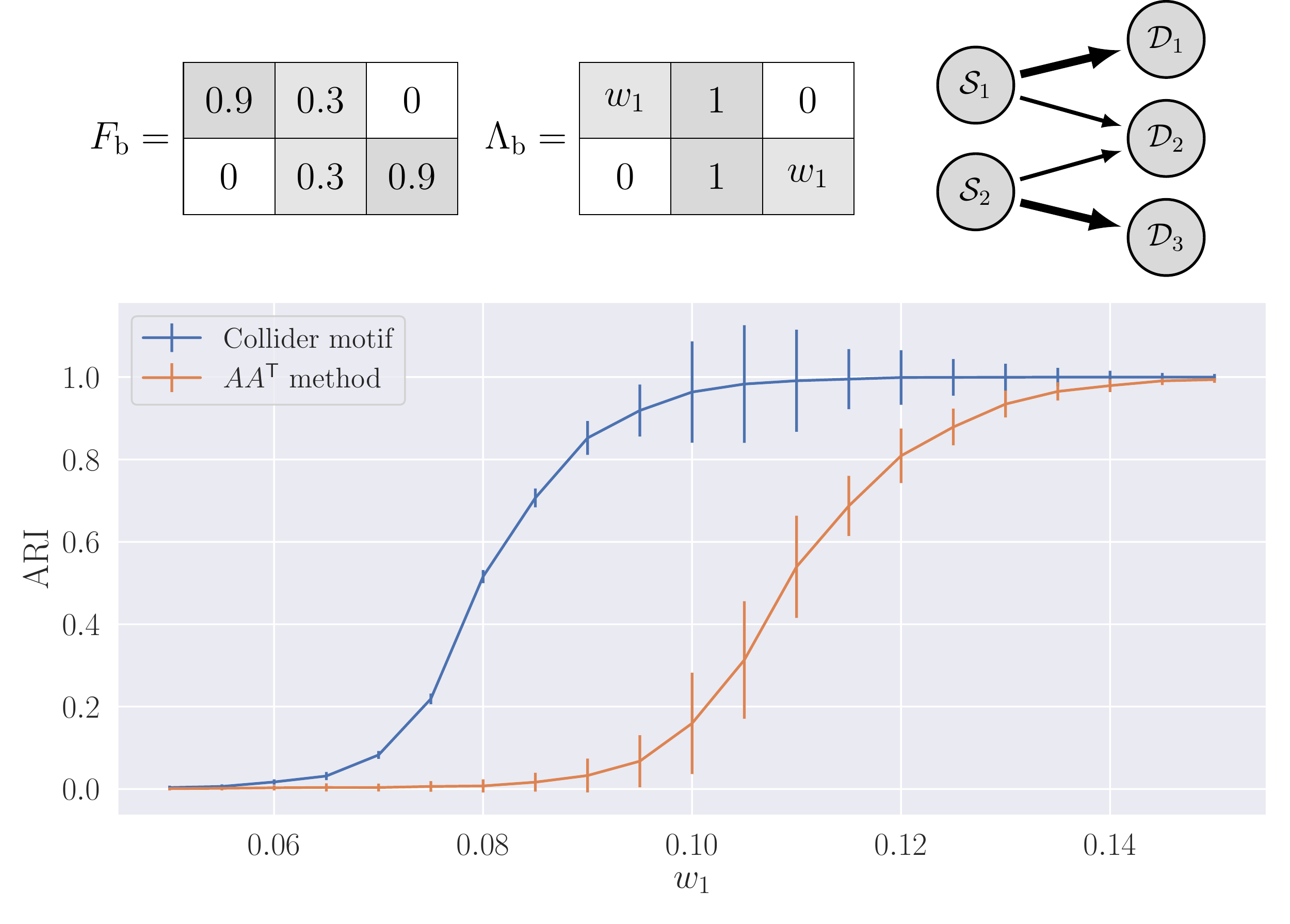}
    \caption{
      \textbf{Top:} BSBM parameters for Example 3 Experiment 2.
      Left panel is the connection matrix,
      middle panel is the weight matrix (see \Cref{eq:weighted} and \Cref{sec:bipartite_sbm} for details),
      right panel is a diagram of the block structure, with arrow widths representing
      connection probabilities.
      \textbf{Bottom:} performance of our bipartite collider method vs. clustering based on $AA^\mathsf{T}$.
      We perform 100 repeats,
      and error bars are one sample standard
      deviation.
    }
    \label{fig:bipartite_motif}
  \end{figure}
}{}

\subsubsection{Experiment 2}

In this final experiment,
we show the advantages of using the collider motif
over simply clustering using the matrix
$A A^\mathsf{T}$, where
$A \in [0,\infty)^{|\ca{S}| \times |\ca{D}|}$
is the weighted bipartite adjacency matrix of the graph
\citep{chessa2014cluster}.
We use a weighted BSBM with
block counts
$k_\ca{S} = 2$,
$k_\ca{D} = 3$,
block sizes
$n_\ca{S}^1 = n_\ca{S}^2 = n_\ca{D}^1 = n_\ca{D}^2 = n_\ca{D}^3 = 200$,
bipartite connection matrix
$F_\mathrm{b} =
\begin{psmallmatrix}
0.9 & 0.3 & 0 \\
0 & 0.3 & 0.9 \\
\end{psmallmatrix}$,
and bipartite weight matrix
$\Lambda_\mathrm{b} =
\begin{psmallmatrix}
w_1 & 1 & 0 \\
0 & 1 & w_1 \\
\end{psmallmatrix}$,
where $w_1$
is a varying parameter.

This model has been constructed such that vertices in
$\ca{S}_1$
are more likely to be connected to
$\ca{D}_1$
than to
$\ca{D}_2$,
and similarly vertices in
$\ca{S}_2$
are more likely to be connected to
$\ca{D}_3$
than to
$\ca{D}_2$.
However, the weights show a reverse preference (for $w_1 < 1$, as in our model),
with larger weights assigned to the
lower-probability edges.
Note that in this regime where weights are small,
there is a significant probability of the Poisson weight variables
being zero,
removing edges which might otherwise have had a high
probability of existing.

The results are shown in
\Cref{fig:bipartite_motif},
which illustrates that our collider method
(\Cref{alg:bipartite_clustering}
with $k_\ca{S} = l_\ca{S} = 2$)
is better at distinguishing the
source vertex clusters
$\ca{S}_1$
and
$\ca{S}_2$
compared to the method based on  $AA^\mathsf{T}$.
Although it must be made clear that the model has been specifically constructed to do this, and that this effect occurs only for a limited parameter set, it gives an example of the different behavior observed when using a product-based weight formulation
(\Cref{sec:weighted_mams}).
\green{We explore this further in real world data in \Cref{sec:languages}}.
For our collider method,
the expected similarity between two source vertices in
the same cluster directly depends on
$w_1$,
while the expected similarity between two source vertices in
different clusters is fixed.
Hence for large enough values of $w_1$,
our collider method performs well.
However, for the $AA^\mathsf{T}$ method,
the expected similarity between two source vertices in
the same cluster depends on
$w_1^2$
(since $AA^\mathsf{T}$ contains squared edge weights),
which in this model is small
(as $w_1 < 1$).
Hence the $AA^\mathsf{T}$ method does not perform so well
on this model.

\ja{}{\pagebreak}

\section{Applications to Real World Data}
\label{sec:realWorld}
This section details the
results of our proposed methodology
on a number of directed networks arising from real world data sets.
We show that weighted edges and motif-based clustering
allow
\green{various}
structures to be uncovered in migration data
with the \textsc{US-Migration} network,
while the \textsc{US-Political-Blogs} network demonstrates how motif-based methods
can control misclassification error for weakly-connected vertices,
\green{how sweep profiles can be used to select a motif,
and how regularization can affect clustering.}
Finally, we consider the bipartite territory-to-language network, \textsc{Unicode-Languages},
and show that
\green{when using our}
bipartite clustering method,
\green{mean-weighted motifs can produce more desirable
clusters than their unweighted or product-weighted counterparts.}

\subsection{\textsc{US-Migration} network} \label{sec:motif_migration}

The first data set is the \textsc{US-Migration} network \citep{census2000}, consisting of data collected during the US Census in 2000.
Vertices represent the 3075 counties in 49 contiguous states (excluding Alaska and Hawaii, and including the District of Columbia).
The $721\,432$ weighted directed edges represent the number of people migrating from county to county, capped at $10 \, 000$ (the 99.9th percentile) to control large entries, as in \cite{DirectedClustImbCuts}.

\wf{\begin{figure}[t]
  \centering
  \includegraphics[width=0.9\textwidth]{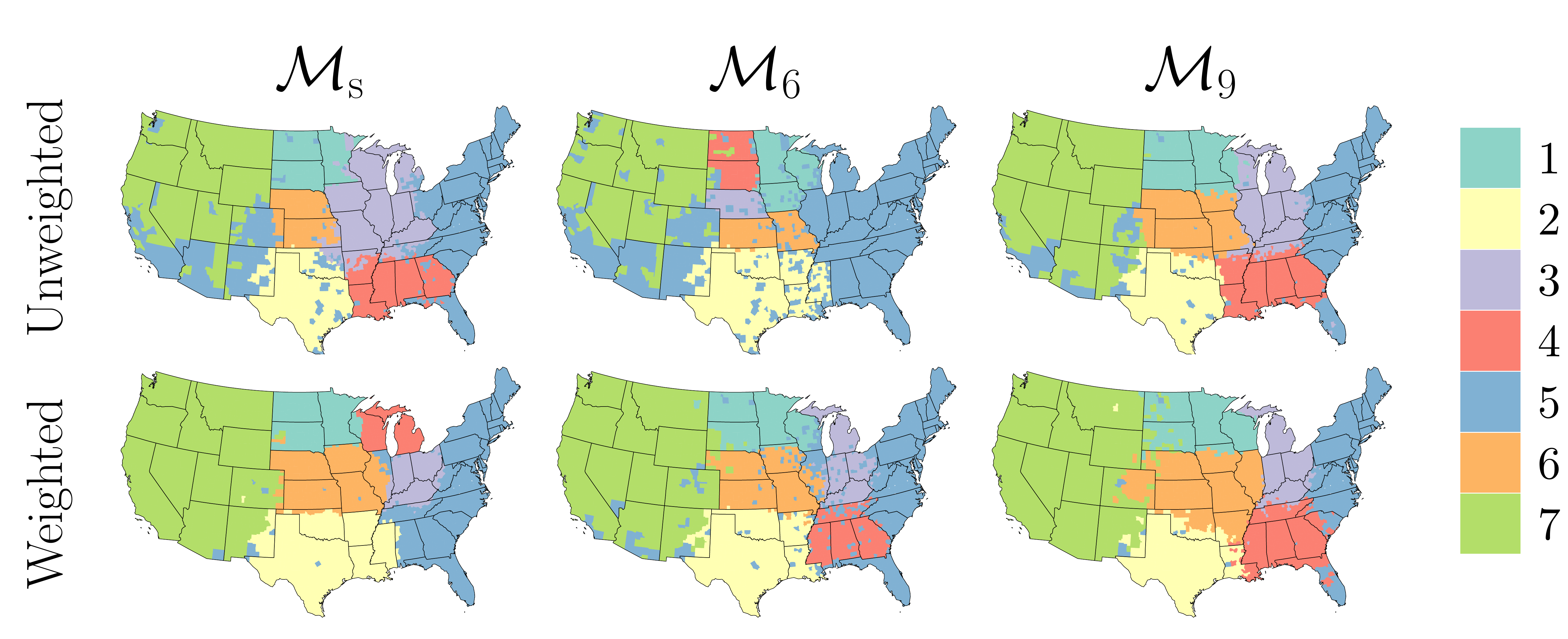}
  \caption{\green{Motif-based clusterings of the \textsc{US-Migration} network, for various functional motifs.
      Each county (node) is colored by its cluster allocation.
      {\bf Top row}: clusterings from the unweighted graph.
      {\bf Bottom row}: clusterings from the weighted graph.
      $\ca{M}_\mathrm{s}$ corresponds to using the symmetrized adjacency matrix.}}
  \label{fig:us_migration}
\end{figure}
}{}

We test the performance of \Cref{alg:motifrwspectclust} with three selected motifs:
$\ca{M}_\mathrm{s}$, $\ca{M}_6$ and $\ca{M}_9$ (see \Cref{fig:motif_definitions_directed}).
$\ca{M}_\mathrm{s}$ gives a spectral clustering method with naive symmetrization.
$\ca{M}_6$ represents a pair of counties exchanging migrants, with both also receiving migrants from a third.
$\ca{M}_9$ is a path of length two, allowing counties to be deemed similar if
they both lie on a heavy-weighted 2-path in the migration network.
\green{These motifs give a representative sample,
including a triangle and a non-triangle motif,
and results with the other motifs
(both functional and structural)
can be found in
Figures \ref{fig:appendix_migration_clusts_func}
and \ref{fig:appendix_migration_clusts_struc}
in \Cref{sec:additionReal}.}

\Cref{fig:us_migration} plots maps of the US,
with counties colored
by the clustering obtained using \Cref{alg:motifrwspectclust} with $k=l=7$.
The top row shows clusterings from the unweighted graph,
while those in the bottom row are from the weighted graph.
The advantage of using weighted graphs is clear,
with the clustering showing better spatial coherence
(compared  with the fragmented blue clusters produced by the unweighted graph).
Furthermore,
\green{as suggested in \Cref{sec:motif_choice}},
the different motifs induce different similarity measures, and thus uncover distinct structures.
For example, \green{weighted} $\ca{M}_6$ and $\ca{M}_9$ both allocate a cluster to the southern states of Mississippi, Alabama, Georgia and Tennessee,
while \green{weighted} $\ca{M}_s$ identifies the northern states of Michigan and Wisconsin.
Also, \green{weighted} $\ca{M}_9$ favors a larger ``central'' region, which includes significant parts of Colorado, Oklahoma, Arkansas and Illinois.

\green{It is also interesting to investigate the cut imbalance ratio
(\Cref{sec:cut_imbalance_ratio})
\citep{DirectedClustImbCuts}
associated with these clusters.
For each pair of clusters, it measures the imbalance of migration flow from one to the other.
\Cref{fig:migration_flow} indicates, for each method,  the first (resp. second) pair of clusters which exhibit the largest (resp. second largest) cut imbalance ratio, with migration mostly occurring from the red counties to the blue counties. We note that the domestic migration report from this census \citep{census2000report}
states that the South had the most in-migration and out-migration, with in-migration accounting for over 60\% of the region's total.
In \Cref{fig:migration_flow}, this is only consistently observed (by coloring the South blue) when using weighted three-node motifs.}
\green{The report also states that out-migration accounts for  over 64\% of migration involving the Northeast, and this
is witnessed only by the motif $\ca{M}_9$ (by coloring the Northeast red).}
\green{It is also worth mentioning that the pattern
of flow from Northeast to South
uncovered by
$\ca{M}_9$ on the unweighted graph (top pair) and the weighted graph (second pair) bears significant resemblance to the topmost imbalanced structure uncovered by the
\textsc{Herm-Sym} algorithm in Figure 16 within \cite{DirectedClustImbCuts}.}

\green{\Cref{fig:migration_heatmaps}
displays the full cut imbalance ratio matrices
(\Cref{sec:cut_imbalance_ratio})
associated with
each method.
The $(i,j)$th matrix entry is positive if there is
more migration flow from cluster $i$ to cluster $j$
than from cluster $j$ to cluster $i$
(yielding an antisymmetric matrix),
and a value of $\pm 1/2$ indicates uni-directional flow.
We see that weighted $\ca{M}_6$ and $\ca{M}_9$
produce clear ``stripes'' of blue and red in the matrices,
where they identify clusters which have
more imbalanced migration.
This again highlights the ability of weighted motifs
to uncover hidden structures
(in this case corresponding to large-scale migration patterns).}

\wf{\begin{figure}[t]
    \centering
    \includegraphics[width=0.7\textwidth]{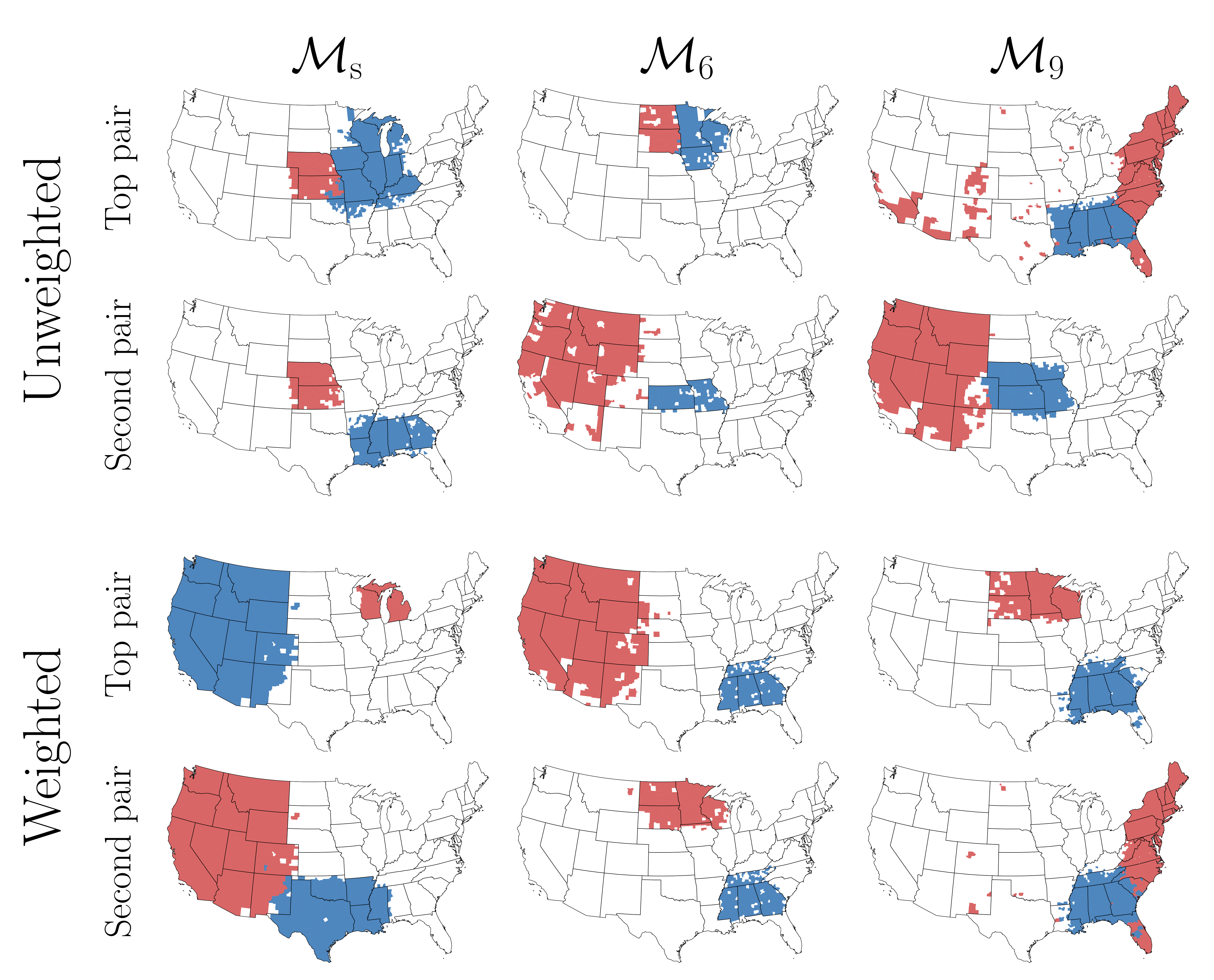}
    \caption{\green{Cut imbalance for functional motif-based clusterings of the \textsc{US-Migration} network, for various motifs.
        For each method, the pairs of clusters with largest
        and second largest cut imbalance ratio are
        colored, with most migration flow occurring from red to blue.
        {\bf Top two rows:} clusterings from the unweighted graph.
        {\bf Bottom two rows:} clusterings from the weighted graph.
        $\ca{M}_\mathrm{s}$ corresponds to using the symmetrized adjacency matrix.}}
    \label{fig:migration_flow}
  \end{figure}
}{}

\wf{\begin{figure}[H]
    \centering
    \includegraphics[width=0.8\textwidth]{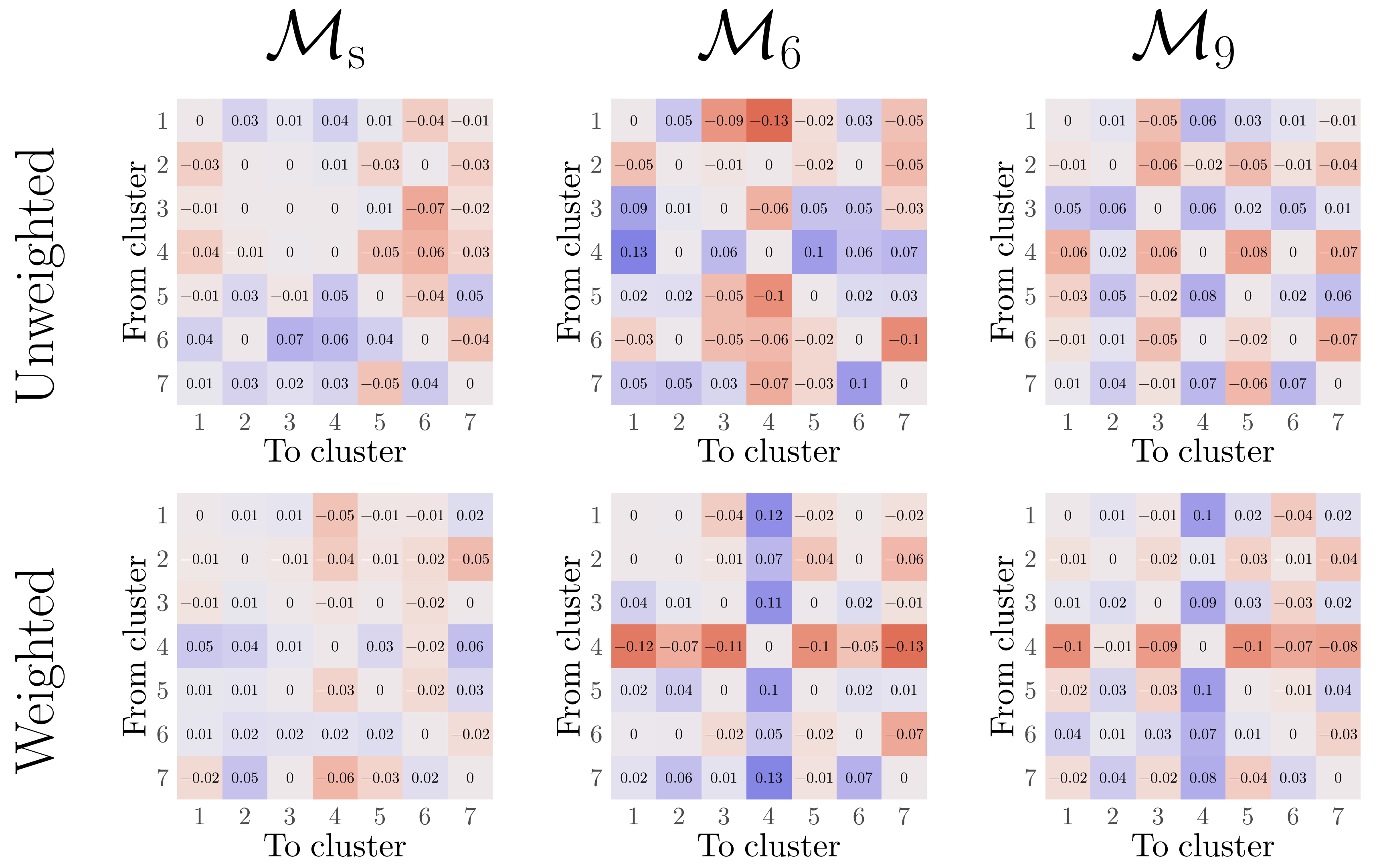}
    \caption{\green{Cut imbalance ratio matrices for functional motif-based clusterings of the \textsc{US-Migration} network, for various motifs.
        Entry $(i,j)$ is positive if most of the flow is in the direction $i \to j$.
        {\bf Top row:} clusterings from the unweighted graph.
        {\bf Bottom row:} clusterings from the weighted graph.
        $\ca{M}_\mathrm{s}$ corresponds to using the symmetrized adjacency matrix.} }
    \label{fig:migration_heatmaps}
  \end{figure}
}{}

\wf{\begin{figure}[t]
  \centering
		\includegraphics[width=0.95\textwidth]{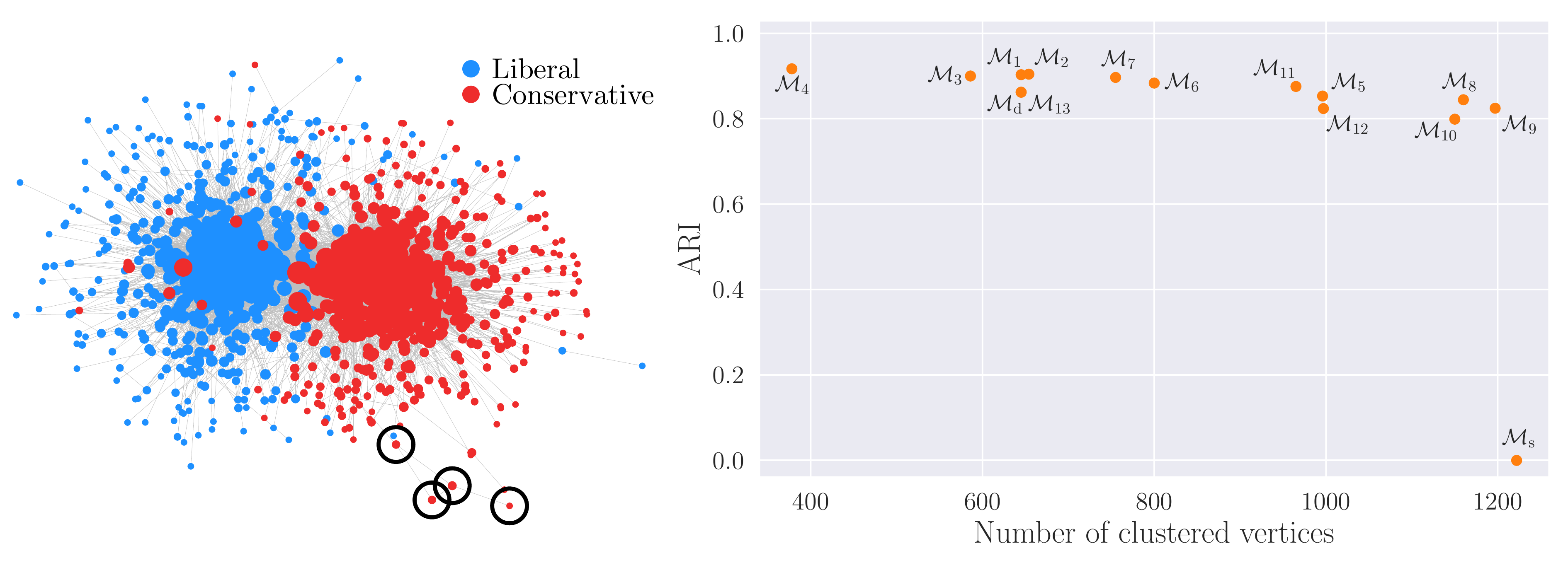}\caption{{\bf Left:} the \textsc{US-Political-Blogs} network.
	{\bf Right:} ARI score against number of clustered vertices, across various motifs.
	}
	\label{fig:polblogs}
\end{figure}
}{}

\wf{\pagebreak}{}

\subsection{\textsc{US-Political-Blogs} network} \label{sec:motif_polblogs}

The \textsc{US-Political-Blogs} network \citep{adamic2005political}
consists of data collected two months before the 2004 US election.
Vertices represent blogs, and are labeled by their political leaning (``liberal'' or ``conservative'').
Weighted directed edges represent the number of citations from one blog to another.
After restricting to the largest connected component,
there are $536$ liberal blogs, $636$ conservative blogs (total $1222$) and $19 \, 024$ edges.
The network is plotted in \Cref{fig:polblogs}.

\green{Firstly we use sweep profiles \citep{shi2000normalized}
with a few selected motifs to motivate the use of motif-based clustering on this network.
To construct a sweep profile,
we order the vertices according to the first non-trivial
eigenvector of the random-walk Laplacian,
and select a splitting point based on minimizing an objective function;
in this case the Ncut score
(see~\Cref{chap:spectral}).
\Cref{fig:polblogs_sweep_profile}
exhibits visibly sharper minima for motifs
$\ca{M}_3$ and $\ca{M}_8$
than for $\ca{M}_s$
(which corresponds to traditional spectral clustering),
indicating that motif-based methods can produce more
well-defined clusters
(for the vertices contained in the largest connected component of the MAM)
than traditional methods on this network.}

\wf{\begin{figure}[H]
    \centering
    \includegraphics[scale=0.32]{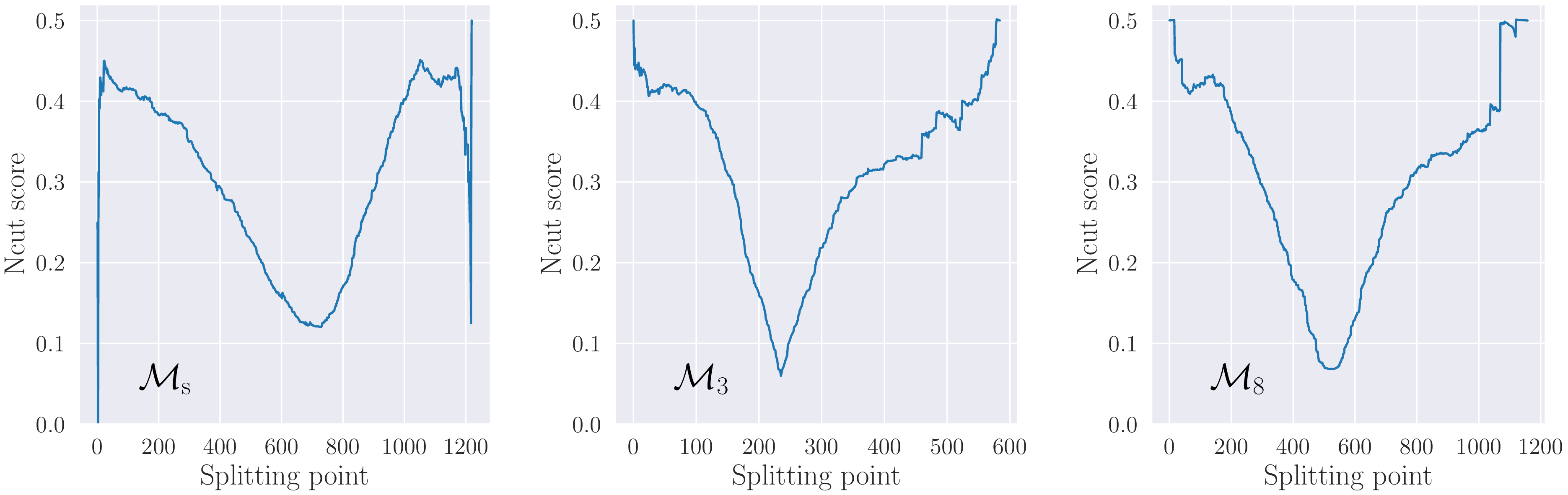}
    \caption{\green{Sweep profiles for different motifs on the \textsc{US-Political-Blogs} network.
        $\ca{M}_\mathrm{s}$ corresponds to using the symmetrized adjacency matrix.}}
    \label{fig:polblogs_sweep_profile}
  \end{figure}
}{}

In fact,
traditional two-cluster spectral clustering performs very poorly on this network,
with one of the clusters containing just four vertices
(circled in \Cref{fig:polblogs}),
which are very weakly connected to the rest of the graph.
However, motif-based clustering performs significantly better.
\green{As discussed in
\Cref{sec:connected_components},
our framework
tries to reduce the misclassification error by only clustering the nodes
which are in the largest connected component of the MAM,
and are therefore in some sense well connected to the graph.}
\green{In this dataset,} the choice of motif determines the trade-off between the number of vertices assigned to clusters,
and the accuracy of those clusters
(see \Cref{fig:polblogs} for a plot of ARI score against number of vertices clustered).
\green{For example,}
motif $\ca{M}_9$ clusters 1197 vertices with an ARI of 0.82,
while the more strongly connected $\ca{M}_4$ only clusters $378$ vertices, with an improved ARI of 0.92.
This is because the largest connected component of the MAM will be smaller for more strongly-connected motifs such as $\ca{M}_4$.

\green{To explore the classification
performance of our
reduced clustering method,
we compare
its effectiveness
with other approaches from the literature.}
\green{\Cref{tab:polblogs_performance} gives the ARI,
number of clustered vertices,
normalized mutual information (NMI) \citep{nguyen2009information},
and percentage error
for various motifs on this network.
We note that both of the motifs $\ca{M}_3$
and $\ca{M}_8$
outperform the NMI score of 0.72 given by
the degree-corrected stochastic blockmodel in
\cite{karrer2011stochastic},
and improve on the percentage error of 4.66\% given by
the normalized spectral clustering procedure in
\cite{li2017divided},
at the expense of not assigning every vertex to a cluster.}

\wf{{
\begin{table}[H]
\centering
\small
\begin{tabular}{|c|c|c|c|c|}

\hline

Motif & ARI & Number of clustered vertices & NMI & Percentage error \\

\hline

$\ca{M}_\mathrm{s}$ &
0.00 &
1222 &
0.00 &
48.2\% \\
$\ca{M}_3$ &
0.90 &
586 &
0.83 &
2.6\% \\
$\ca{M}_8$ &
0.84 &
1160 &
0.75 &
4.1\% \\

\hline

\end{tabular}
\vspace{2mm}
\caption{\green{Performance of motif-based clustering on the \textsc{US-Political-Blogs} network.}}
\label{tab:polblogs_performance}
\end{table}
}
}{}

\green{We further compare our
proposed method
to an alternative approach
for dealing with disconnected vertices;
namely regularization.}
\green{We follow the method suggested by
\cite{joseph2016impact}
and
\cite{zhang2018understanding}
using the
regularized Laplacian
$ L^\tau_\mathrm{rw} = I - (D + \tau I)^{-1} (G + \tau {\bf 1} / n)$,
where ${\bf 1}$ is a matrix of all ones
and $\tau$ is a tuning parameter.
\Cref{tab:polblogs_reg}
shows how
using this regularized Laplacian
with $\tau \sim \bar{d}$
(the mean weighted vertex degree),
can improve the ARI performance of traditional
(top row)
and motif-based (second and third rows)
spectral clustering
while still assigning all of the vertices to clusters.
However this regularization method is unable to reach the ARI scores of over
$0.9$ achieved by our proposed method of
motif-based clustering on the largest connected component
of the MAM
(\Cref{tab:polblogs_performance}).
Furthermore,
regularization is sensitive to the
choice of tuning parameter $\tau$,
and the default setting of
$\tau=\bar{d}$
suggested by
\cite{zhang2018understanding}
and
\cite{qin2013regularized}
performs rather poorly here.}

\wf{{
\setlength{\tabcolsep}{1.8mm}
\begin{table}[H]
\centering
\scriptsize
\begin{tabular}{|c|c|c|c|c|c|c|c|c|c|c|c|c|c|}

\hline
$\nicefrac{\tau}{\bar{d}}$ & 10 & 5 & 2 & 1 & 0.5 & 0.2 & 0.1 & 0.05 & 0.02 & 0.01 & 0.005 & 0.002 & 0.001 \\
\hline
$\ca{M}_\mathrm{s}$
& 0.08
& 0.14
& 0.28
& 0.36
& 0.49
& 0.81
& 0.80
& 0.81
& 0.81
& 0.82
& 0.82
& 0.00
& 0.00
\\
$\ca{M}_\mathrm{3}$
& 0.00
& 0.00
& 0.01
& 0.02
& 0.04
& 0.10
& 0.14
& 0.17
& 0.24
& 0.25
& 0.25
& 0.00
& 0.00
\\
$\ca{M}_\mathrm{8}$
& 0.05
& 0.07
& 0.13
& 0.28
& 0.49
& 0.58
& 0.62
& 0.65
& 0.71
& 0.75
& 0.80
& 0.80
& 0.80
\\
\hline

\end{tabular}
\caption{\green{ARI scores for regularized traditional and motif-based spectral clustering on the \textsc{US-Political-Blogs} network.}}
\label{tab:polblogs_reg}
\end{table}
}
}{}

\subsection{\textsc{Unicode-Languages} bipartite network}
\label{sec:languages}

The \textsc{Unicode-Languages} network \citep{konect:unicodelang} is a bipartite network
consisting of data collected in 2014 on languages spoken around the world.
Source vertices are territories, and destination vertices are languages.
Weighted edges from territory to language indicate the number of inhabitants
\citep{geonames}
in that territory who speak the specified language.
After removing territories with fewer than one million inhabitants, and
languages with fewer than one million speakers,
there are $155$ territories, $270$ languages and $705$ edges remaining.
We then restrict the network to its largest connected component,
and in doing so remove the territories of
Laos, Norway and Timor-Leste,
and the languages of
Lao, Norwegian Bokm{\aa}l and Norwegian Nynorsk.
We test \Cref{alg:bipartite_clustering} with parameters $k_\ca{S} = l_\ca{S} = k_\ca{D} = l_\ca{D} = 6$ on this network
\green{(such that each cluster in
\Cref{tab:bipartite_languages_source_clusters}
contains at least $1\%$ of the world's population),}
comparing three different motif weighting schemes:
unweighted motifs,
mean-weighted motifs
and product-weighted motifs,
as in \Cref{sec:weighted_mams}.
For the source vertices,
\Cref{fig:bipartite_languages_map} plots maps of the world with territories colored by the recovered clustering.

\green{Focusing first
purely on the
clusters produced by
mean-weighted motifs,
we observe
balanced clusters with good spatial coherence
(middle panel in \Cref{fig:bipartite_languages_map}),
from which some language groups can be easily recognized.}
The top 20 territories (by population) in each cluster
for this mean-weighted scheme
are given in \Cref{tab:bipartite_languages_source_clusters}.
Cluster~1 is by far the largest cluster,
and includes a wide variety of territories.
Cluster~2 contains several Persian-speaking and African French-speaking territories,
while Cluster~3 mostly captures Spanish-speaking territories in the Americas.
Cluster~4 includes the Slavic territories of Russia and some of its neighbors.
Cluster~5 covers China, Hong Kong, Mongolia
and some of South-East Asia.
Cluster~6 is the smallest cluster and contains only Japan and the Koreas.

\wf{\renewcommand{\arraystretch}{1}
\begin{table}[H]
  \centering
  \scriptsize
	\begin{tabular}{ |@{ }c@{ }|@{ }c@{ }|@{ }c@{ }|@{ }c@{ }|@{ }c@{ }|@{ }c@{ }| }
		\hline
		\rule{0pt}{1.2em}
    Cluster 1 & Cluster 2 & Cluster 3 & Cluster 4 & Cluster 5 & Cluster 6  \\[0.1cm]
    \hline \rule{0pt}{1.2em} India & Iran & Mexico & Russia & China & Japan  \\
    United States & DR Congo & Colombia & Ukraine & Indonesia & S. Korea  \\
    Brazil & Afghanistan & Argentina & Uzbekistan & Vietnam & N. Korea  \\
    Pakistan & Saudi Arabia & Peru & Belarus & Malaysia &   \\
    Bangladesh & Syria & Venezuela & Tajikistan & Taiwan &   \\
    Nigeria & C\^ote d'Ivoire & Ecuador & Kyrgyzstan & Cambodia &   \\
    Philippines & Burkina Faso & Guatemala & Turkmenistan & Hong Kong &   \\
    Ethiopia & Niger & Cuba & Georgia & Singapore &   \\
    Germany & Mali & Bolivia & Moldova & Panama &   \\
    Egypt & Senegal & Paraguay & Latvia & Mongolia &   \\
    Turkey & Tunisia & El Salvador & Estonia &  &   \\
    Thailand & Chad & Nicaragua &  &  &   \\
    France & Guinea & Costa Rica &  &  &   \\
    United Kingdom & Somalia & Uruguay &  &  &   \\
    Italy & Burundi & Eq. Guinea &  &  &   \\
    Myanmar & Haiti &  &  &  &   \\
    South Africa & Benin &  &  &  &   \\
    Spain & Azerbaijan &  &  &  &   \\
    Tanzania & Togo &  &  &  &   \\
    Kenya & Libya &  &  &  &   \\
    $\cdots$ & $\cdots$ &  &  &  &   \\
    $|$Cluster 1$|$ = 87 & $|$Cluster 2$|$ = 29 & $|$Cluster 3$|$ = 15
    & $|$Cluster 4$|$ = 11 & $|$Cluster 5$|$ = 10 & $|$Cluster 6$|$ = 3 \\[0.1cm]
		\hline
	\end{tabular}
  \vspace{2mm}
  \caption{Clustering the territories from the \textsc{Unicode-Languages} network
  \green{using the mean-weighted collider motif.}}
  \label{tab:bipartite_languages_source_clusters}
\end{table}
}{}

For the destination vertices, we present the six clusters obtained by \Cref{alg:bipartite_clustering} with mean-weighted motifs.
\Cref{tab:bipartite_languages_dest_clusters} contains the top 20 languages (by number of speakers) in each cluster.
Cluster~1 is the largest cluster and contains some European languages and dialects of Arabic.
Cluster~2 is also large and includes English, as well as several South Asian languages.
Cluster~3 consists of many indigenous African languages.
Cluster~4 captures languages from South-East Asia, mostly spoken in Indonesia and Malaysia.
Cluster~5 identifies several varieties of Chinese and a few other Central and East Asian languages.
Cluster~6 captures more South-East Asian languages, this time from Thailand, Myanmar and Cambodia.

\wf{\renewcommand{\arraystretch}{1}
\begin{table}[H]
  \centering
  \scriptsize
	\begin{tabular}{ |@{ }c@{ }|@{ }c@{ }|@{ }c@{ }|@{ }c@{ }|@{ }c@{ }|@{ }c@{ }| }
		\hline
		\rule{0pt}{1.2em}
    Cluster 1 & Cluster 2 & Cluster 3 & Cluster 4 & Cluster 5 & Cluster 6 \\[0.1cm]
    \hline \rule{0pt}{1.2em} Spanish & English & Swahili & Indonesian & Chinese & Thai  \\
    Arabic & Hindi & Kinyarwanda & Javanese & Wu Chinese & N.E. Thai  \\
    Portuguese & Bengali & Somali & Malay & Korean & Khmer  \\
    French & Urdu & Luba-Lulua & Sundanese & Xiang Chinese & N. Thai  \\
    Russian & Punjabi & Kikuyu & Madurese & Hakka Chinese & S. Thai  \\
    Japanese & Telugu & Congo Swahili & Minangkabau & Minnan Chinese & Shan  \\
    German & Marathi & Luyia & Betawi & Gan Chinese & Pattani Malay  \\
    Turkish & Vietnamese & Ganda & Balinese & Kazakh &   \\
    Persian & Tamil & Luo & Buginese & Uighur &   \\
    Italian & Lahnda & Sukuma & Banjar & Sichuan Yi &   \\
    Egyptian Arabic & Filipino & Kalenjin & Achinese & Mongolian &   \\
    Polish & Gujarati & Lingala & Sasak & Zhuang &   \\
    Nigerian Pidgin & Kannada & Nyankole & Makasar & Tibetan &   \\
    Ukrainian & Pushto & Gusii & Lampung Api &  &   \\
    Dutch & Malayalam & Kiga & Rejang &  &   \\
    Algerian Arabic & Oriya & Soga &  &  &   \\
    Moroccan Arabic & Burmese & Luba-Katanga &  &  &   \\
    Hausa & Bhojpuri & Meru &  &  &   \\
    Azerbaijani & Amharic & Teso &  &  &   \\
    Uzbek & Oromo & Nyamwezi &  &  &   \\
    $\cdots$ & $\cdots$ & $\cdots$ &  &  &   \\
    $|$Cluster 1$|$ = 120 & $|$Cluster 2$|$ = 90 & $|$Cluster 3$|$ = 25
    & $|$Cluster 4$|$ = 15 & $|$Cluster 5$|$ = 13 & $|$Cluster 6$|$ = 7 \\[0.1cm]
		\hline
	\end{tabular}
  \vspace{2mm}
  \caption{Clustering the languages from the \textsc{Unicode-Languages} network
  \green{using the mean-weighted expander motif.}}
  \label{tab:bipartite_languages_dest_clusters}
\end{table}
}{}

\green{Next, we
  return to the source vertices (territories) and compare the mean-weighted scheme
  with the unweighted and product-weighted schemes, as illustrated in
  \Cref{fig:bipartite_languages_map}. Unlike the mean-weighted approach, the unweighted partition fails to identify the Chinese-speaking territories, which is likely to be related to the fact that this scheme is unable to account for the large number of speakers in China. Further, it does not identify Mexico and Argentina with other Spanish-speaking
  territories in the Americas.
  Finally,
the product-weighted motifs make
  no distinctions at all between territories in the Americas,
  Africa and Western Europe,
  and instead identify a few much smaller clusters.}

\wf{\begin{figure}[H]
    \centering
    \includegraphics[width=0.73\textwidth]{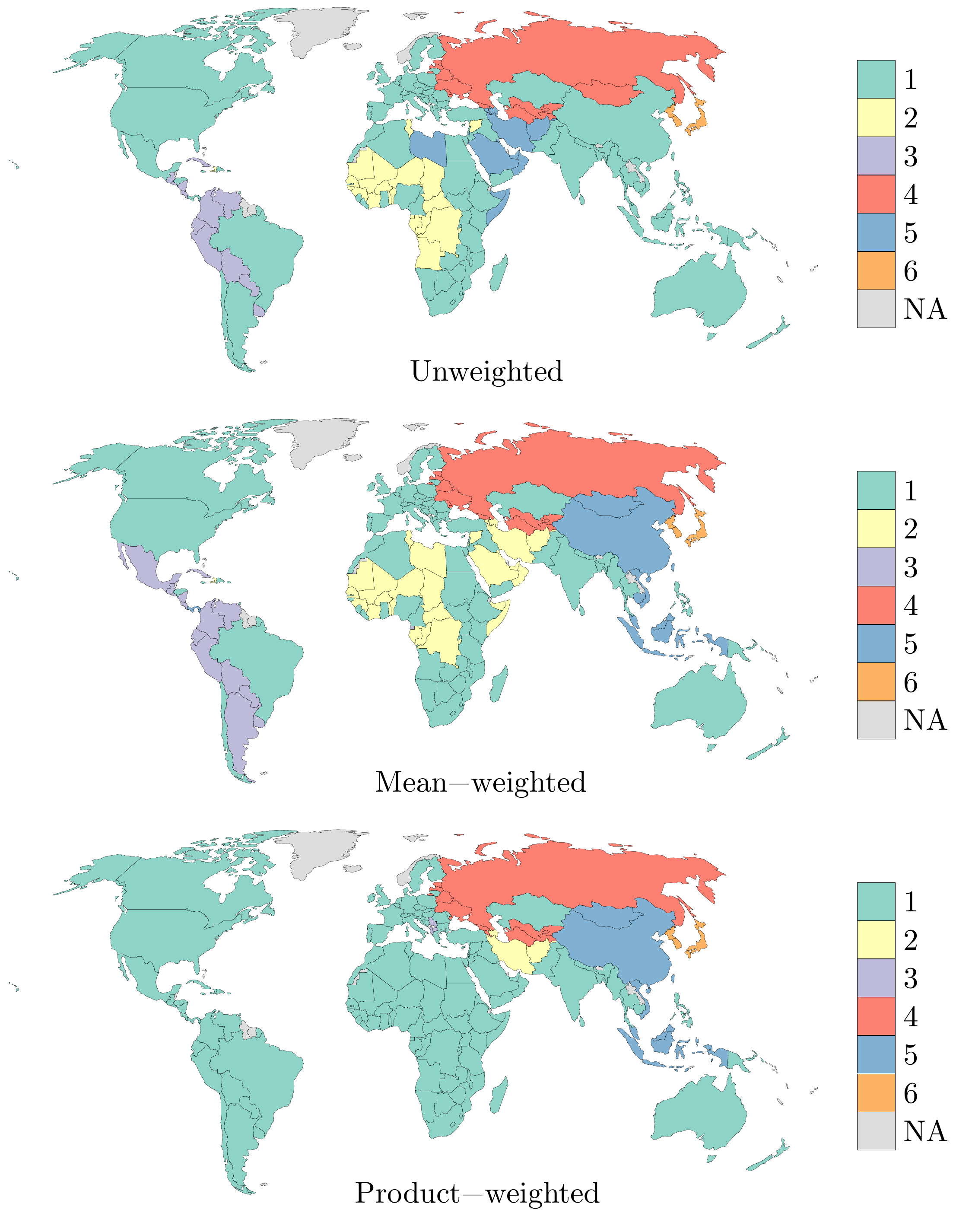}
    \caption{Clustering the territories from the \textsc{Unicode-Languages} network,
      \green{using the collider motif.
        {\bf Top:} clustering the unweighted graph.
        {\bf Middle:} motifs weighted using mean edge weight.
        {\bf Bottom:} motifs weighted using product of edge weights.}
    }
    \label{fig:bipartite_languages_map}
  \end{figure}
}{}
 \ja{}{\pagebreak}

\section{Conclusion and Future Work} \label{chap:conclusions}

\paragraph{Contribution}
We have
introduced
generalizations of MAMs to weighted directed graphs
and presented new matrix-based formulae for calculating them,
which are easy to implement and scalable.
By leveraging the popular random-walk spectral clustering algorithm,
we have proposed motif-based techniques for clustering both weighted directed graphs and weighted
bipartite graphs.
We demonstrated
using
synthetic data examples that accounting for edge weights and higher-order structure can be essential for good cluster recovery in directed and bipartite graphs,
and that different motifs can uncover different but equally insightful clusters.
Applications to real world networks have shown that weighted motif-based clustering can reveal a variety of different structures,
and that it can reduce misclassification rate,
\green{performing favorably
in comparison with other methods from the literature.}

\paragraph{Future work}

There are limitations to our proposed methodology,
which raise a number of research avenues to explore.
While our matrix-based formulae for MAMs are simple to implement and scalable,
there is scope for further computational improvement.
As mentioned in~\citep{benson2016higher}, fast triangle enumeration  algorithms~\citep{demeyer2013ISMA,wernicke2006efficient,wernicke2006fanmod}
may offer superior performance at scale for triangular motifs, but MAMs for motifs with missing edges are inherently denser.
The implementation could also be optimized to further exploit any sparsity structures in the network, rather than relying on general-purpose linear algebra libraries.
\green{More recent work on
combinatorial graphlet counting
\citep{ahmed2015efficient}
or vertex orbit counting
\citep{pashanasangi2020efficiently}
could also potentially yield further computational improvements.}
\green{Another shortcoming of our matrix-based formulae is that,
unlike motif detection algorithms such as
\textsc{FANMOD}
\citep{wernicke2006fanmod},
they do not easily extend to motifs on four or more vertices.
While we believe that this extension
is in principle possible using tensor
operations rather than matrix operations,
it is currently beyond the scope of the work.}

We have seen that choosing the correct motif
for clustering a network can play a crucial role
(\Cref{sec:syntheticexample2}),
\green{and while we have suggested some methods for doing so
(\Cref{sec:motif_choice}),
it may be that
other techniques
can also be used to
guide the motif selection process.}
\green{We would also
be interested in seeing more
approaches to dealing with disconnected MAMs,
perhaps involving some criterion for when a
connected component is ``large enough'' to be
worth keeping,
or a deeper exploration of regularization methods
for motif-based spectral clustering.}
Extensions to our methodology could involve further
analysis of the differences between
\green{functional and structural MAMs,
between different types of Laplacians
\citep{von2007tutorial},
or between different weighted stochastic graph models,
perhaps following an exponential family method
\citep{aicher2013adapting}.
Also, connections between motif-based clustering and
hypergraph partitioning \citep{li2017inhomogeneous, li2018submodular, veldt2020minimizing}
would be worth further investigation.}

\green{Further experimental work
would also be interesting.}
We would like to explore the utility of our framework on additional real data,
and suggest that
collaboration networks such as~\cite{snap:astro},
transportation networks such as the airports network in
\cite{frey2007clustering},
and bipartite preference networks such as~\cite{icon:movie}
could be interesting.
Direct comparison of results with other state-of-the-art clustering methods from the literature  could also be insightful; suitable benchmarks for performance include
the Hermitian matrix-based clustering method in~\citep{DirectedClustImbCuts},
the Motif-based Approximate Personalized PageRank (MAPPR) in~\citep{yin2017local}, and
\textsc{Tectonic} from~\citep{tsourakakis2017scalable}.
Other established and fast methods from the literature on directed networks include
{\sc SaPa}~\citep{satuluri2011sym} and {\sc DI-SIM}~\citep{rohe2016co}, which respectively perform a degree-discounted symmetrization; and a combined row and column clustering into four sets using the concatenation of the left and right singular vectors.

Further potential extensions of our work pertain to the detection of overlapping clusters, allowing nodes to belong to multiple groups~\citep{li2016motifOverlapping}, to the detection of core--periphery structures in directed graphs~\citep{elliott2019coreperiphery}, as well as implications for existing methods on network comparison~\citep{wegner2017identifying}, anomaly detection in directed networks \citep{elliott2019anomaly}, and ranking from a sparse set of noisy pairwise comparisons~\citep{syncRank}.
Furthermore, incorporating our pipeline into deep learning architectures is another avenue worth exploring, building on the recent \textsc{MotifNet} architecture~\citep{monti2018motifnet}, a graph convolutional neural network for directed graphs, or  \textsc{SiGAT}~\citep{huang2019signedMotifs}, another recent graph neural network which incorporates graph motifs in the context of signed networks which contain both positive and negative links~\citep{sponge}.

\section*{Acknowledgments}
Andrew Elliott and Mihai Cucuringu acknowledge support from the
EPSRC grant EP/N510129/1 at The Alan Turing Institute and Accenture PLC.
We would like to thank
Lyuba Bozhilova
for her feedback on this work.
We would also like to thank
Microsoft for
their generous donation of Azure credits to The Alan Turing Institute which made this work possible.

\section*{Availability of data and materials}
All data sets are cited and are
publicly available online.
Source code in R and Python is available at
\href{https://github.com/wgunderwood/motifcluster}{\texttt{https://github.com/wgunderwood/motifcluster}}.

\bibliographystyle{spbasic}
\bibliography{refs}

\appendix
\newpage
\section{Motif Adjacency Matrix Formulae} \label{chap:appendix_matrices}

We give explicit matrix-based formulae for mean-weighted motif adjacency matrices
$M$ for all simple motifs $\ca{M}$
on at most three vertices,
along with the anchored motifs $\ca{M}_\mathrm{coll}$ and $\ca{M}_\mathrm{expa}$.
Entry-wise (Hadamard) products are denoted by $\circ$.
\Cref{tab:motif_adj_mat_table}
gives our dense formulation of functional MAMs.
For the dense formulation of
structural MAMs,
simply replace $J_\mathrm{n}$,
$J$, and $G$ with $J_0$, $J_\mathrm{s}$, and $G_\mathrm{s}$ respectively.
For the sparse formulation of functional MAMs,
replace
$J_\mathrm{n}$
with
${\bf 1} - I$,
and expand and simplify the expressions
as in
\Cref{sec:computational_analysis}.
For the sparse formulation of structural MAMs,
replace
$J$ and $G$ with $J_\mathrm{s}$ and $G_\mathrm{s}$ respectively.
Also replace $J_0$ with ${\bf 1} - (I + \tilde{J})$,  where
$\tilde{J} = J_\mathrm{s} + J_\mathrm{s}^\mathsf{T} + J_\mathrm{d}$,
and expand and simplify the expressions.

\wf{\renewcommand{\arraystretch}{2.5}
\setlength\tabcolsep{1.8pt}
\vspace*{5mm}
\begin{table}[ht]

	\centering
	\fontsize{07}{07}\selectfont
	\begin{tabular}{ |c|c|c|c| }

		\hline

		Motif & $C$ & $C'$ & $M^\mathrm{func}$ \\

		\hline

		$\ca{M}_\mathrm{s}$ & & & $G + G^\mathsf{T}$ \\

		\hline

		$\ca{M}_\mathrm{d}$ & & & $\frac{1}{2} G_\mathrm{d}$ \\

		\hline

		$\ca{M}_1$ & $J^\mathsf{T} \circ (J G) + J^\mathsf{T} \circ (G J) + G^\mathsf{T} \circ (J J)$ & & $\frac{1}{3} \big(C + C^\mathsf{T}\big)$ \\

		\hline

		$\ca{M}_2$ & \rule{0pt}{2.7em}$\displaystyle
		\begin{aligned}
       					& J^\mathsf{T} \circ (J_\mathrm{d} G)
       					+ J^\mathsf{T} \circ (G_\mathrm{d} J)
       					+ G^\mathsf{T} \circ (J_\mathrm{d} J)
       					\\ &
       					+ J^\mathsf{T} \circ (J G_\mathrm{d})
       					+ J^\mathsf{T} \circ (G J_\mathrm{d})
       					+ G^\mathsf{T} \circ (J J_\mathrm{d})
       					\\ &
       					+ J_\mathrm{d} \circ (J G)
       					+ J_\mathrm{d} \circ (G J)
       					+ G_\mathrm{d} \circ (J J)
    				\end{aligned}$\rule[-2em]{0pt}{1em} & & $\frac{1}{4} \big(C + C^\mathsf{T}\big)$ \\

    	\hline

		$\ca{M}_3$ & \rule{0pt}{2.7em}$\displaystyle\begin{aligned}
						& J \circ (J_\mathrm{d} G_\mathrm{d})
						+ J \circ (G_\mathrm{d} J_\mathrm{d})
						+ G \circ (J_\mathrm{d} J_\mathrm{d})
       					\\ &
						+ J_\mathrm{d} \circ (J_\mathrm{d} G)
						+ J_\mathrm{d} \circ (G_\mathrm{d} J)
						+ G_\mathrm{d} \circ (J_\mathrm{d} J)
						\\ &
						+ J_\mathrm{d} \circ (J G_\mathrm{d})
						+ J_\mathrm{d} \circ (G J_\mathrm{d})
						+ G_\mathrm{d} \circ (J J_\mathrm{d})
					\end{aligned}$\rule[-2em]{0pt}{1em} & & $\frac{1}{5} \big(C + C^\mathsf{T}\big)$ \\

		\hline

		$\ca{M}_4$ & $ J_\mathrm{d} \circ (J_\mathrm{d} G_\mathrm{d}) + J_\mathrm{d} \circ (G_\mathrm{d} J_\mathrm{d}) + G_\mathrm{d} \circ (J_\mathrm{d} J_\mathrm{d}) $ & & $ \frac{1}{6} C$ \\

		\hline

		$\ca{M}_5$ & \rule{0pt}{2.7em}$\displaystyle\begin{aligned}
						& J \circ (J G)
						+ J \circ (G J)
						+ G \circ (J J)
						\\ &
						+ J \circ (J G^\mathsf{T})
						+ J \circ (G J^\mathsf{T})
						+ G \circ (J J^\mathsf{T})
						\\ &
						+ J \circ (J^\mathsf{T} G)
						+ J \circ (G^\mathsf{T} J)
						+ G \circ (J^\mathsf{T} J)
					\end{aligned}$\rule[-2em]{0pt}{1em} & & $\frac{1}{3} \big(C + C^\mathsf{T}\big)$ \\

		\hline

		$\ca{M}_6$ & $J \circ (J G_\mathrm{d}) + J \circ (G J_\mathrm{d}) + G \circ (J J_\mathrm{d}) + J_\mathrm{d} \circ (J^\mathsf{T} G)$ & $G_\mathrm{d} \circ (J^\mathsf{T} J)$ & $\frac{1}{4} \big(C + C^\mathsf{T} + C' \big)$ \\

		\hline

		$\ca{M}_7$ & $J \circ (J_\mathrm{d} G) + J \circ (G_\mathrm{d} J) + G \circ (J_\mathrm{d} J)$ &
		\rule{0pt}{2.7em}$\displaystyle\begin{aligned}
		J_\mathrm{d} &\circ (J G^\mathsf{T}) + J_\mathrm{d} \circ (G J^\mathsf{T})
		\\
		& + G_\mathrm{d} \circ (J J^\mathsf{T})
		\end{aligned}$\rule[-2em]{0pt}{1em} &
		$ \frac{1}{4} \big(C + C^\mathsf{T} + C' \big)$ \\

		\hline

		$\ca{M}_8$ & $J \circ (G J_\mathrm{n}) + G \circ (J J_\mathrm{n})$ & $J_\mathrm{n} \circ (J^\mathsf{T} G) + J_\mathrm{n} \circ (G^\mathsf{T} J)$ & $\frac{1}{2} \big(C + C^\mathsf{T} + C' \big)$ \\

		\hline

		$\ca{M}_9$ & \rule{0pt}{1.9em}$\displaystyle\begin{aligned}
						& J \circ (J_\mathrm{n} G^\mathsf{T}) + G \circ (J_\mathrm{n} J^\mathsf{T}) + J_\mathrm{n} \circ (J G) \\
						& + J_\mathrm{n} \circ (G J) + J \circ (G^\mathsf{T} J_\mathrm{n}) + G \circ (J^\mathsf{T} J_\mathrm{n})
					\end{aligned}$\rule[-1.3em]{0pt}{1em} & & $\frac{1}{2} \big(C + C^\mathsf{T}\big)$ \\

		\hline

		$\ca{M}_{10}$ & $J \circ (J_\mathrm{n} G) + G \circ (J_\mathrm{n} J)$ & $J_\mathrm{n} \circ (J G^\mathsf{T}) + J_\mathrm{n} \circ (G J^\mathsf{T})$ & $\frac{1}{2} \big(C + C^\mathsf{T} + C' \big)$ \\

		\hline

		$\ca{M}_{11}$ & \rule{0pt}{1.9em}$\displaystyle\begin{aligned}
						& J_\mathrm{d} \circ (G J_\mathrm{n}) + G_\mathrm{d} \circ (J J_\mathrm{n}) + J_\mathrm{n} \circ (J_\mathrm{d} G) \\
						&  + J_\mathrm{n} \circ (G_\mathrm{d} J) + J \circ (G_\mathrm{d} J_\mathrm{n}) + G \circ (J_\mathrm{d} J_\mathrm{n})
					\end{aligned}$\rule[-1.3em]{0pt}{1em} & & $\frac{1}{3} \big(C + C^\mathsf{T}\big)$ \\

		\hline

	                  $\ca{M}_{12}$ & \rule{0pt}{1.9em}$\displaystyle\begin{aligned}
                                                  & J_\mathrm{d} \circ (J_\mathrm{n} G) + G_\mathrm{d} \circ (J_\mathrm{n} J) + J_\mathrm{n} \circ (J G_\mathrm{d}) \\
                                                  & + J_\mathrm{n} \circ (G J_\mathrm{d}) + J \circ (J_\mathrm{n} G_\mathrm{d}) + G \circ (J_\mathrm{n} J_\mathrm{d})
                                          \end{aligned}$\rule[-1.3em]{0pt}{1em} & & $ \frac{1}{3} \big(C + C^\mathsf{T}\big)$ \\

		\hline

		$\ca{M}_{13}$ & $J_\mathrm{d} \circ (G_\mathrm{d} J_\mathrm{n}) + G_\mathrm{d} \circ (J_\mathrm{d} J_\mathrm{n}) + J_\mathrm{n} \circ (J_\mathrm{d} G_\mathrm{d})$ & & $\frac{1}{4} \big(C + C^\mathsf{T} \big)$ \\

		\hline

		$\ca{M}_\mathrm{coll}$ & $J_\mathrm{n} \circ (J G^\mathsf{T})$ & & $\frac{1}{2} \big( C + C^\mathsf{T} \big)$ \\

		\hline

		$\ca{M}_\mathrm{expa}$ & $J_\mathrm{n} \circ (J^\mathsf{T} G)$ & & $\frac{1}{2} \big( C + C^\mathsf{T} \big)$ \\

		\hline

	\end{tabular}
  \vspace{2mm}
	\caption{\normalsize Dense formulation of mean-weighted functional motif adjacency matrices.}
	\label{tab:motif_adj_mat_table}
	\vspace{-0.5cm}
\end{table}
}{}
 \clearpage
\section{Proofs}\label{chap:appendix_proofs}

\begin{prf}[\Cref{prop:motif_adj_matrix_formula}, MAM formula] \label{proof:motif_adj_matrix_formula}

Consider
\Cref{eq:mamformulafunc}.
We sum over functional instances $\ca{M} \cong \ca{H} \leq \ca{G}$ such that $\{i,j\} \in \ca{A(H)}$.
This is equivalent to summing over $\{k_2, \ldots, k_{m-1}\} \subseteq \ca{V}$ and $\sigma \in S_\ca{M,A}^\sim$, such that $k_u$ are all distinct and

\[ (u,v) \in \ca{E_M} \implies (k_{\sigma u}, k_{\sigma v}) \in \ca{E}\,. \qquad (\dagger) \]

\noindent This is because the vertex set $\{k_2, \ldots, k_{m-1}\} \subseteq \ca{V}$ indicates which vertices are present in the instance $\ca{H}$, and $\sigma$ describes the mapping from $\ca{V_M}$ onto those vertices: $u \mapsto k_{\sigma u}$. We take $\sigma \in S_\ca{M,A}^\sim$ to ensure that $\{i,j\} \in \ca{A(H)}$ (since $i=k_1, \ j=k_m$), and that instances are counted exactly once.
The condition $(\dagger)$ is to check that $\ca{H}$ is a functional instance of $\ca{M}$ in $\ca{G}$. Hence

\begin{align*}
	M^\mathrm{func}_{ij} &= \frac{1}{|\ca{E_M}|} \sum_{\ca{M} \cong \ca{H} \leq \ca{G}} \bb{I} \big\{ \{i,j\} \in \ca{A}(\ca{H}) \big\} \sum_{e \in \ca{E_H}} W(e) \\
&=  \frac{1}{|\ca{E_M}|} \sum_{\{ k_2, \ldots, k_{m-1} \}} \sum_{\sigma \in S_\ca{M,A}^\sim} \bb{I} \big\{ k_u \textrm{ all distinct}, \, (\dagger) \big\} \sum_{e \in \ca{E_H}} W(e)\,.
\end{align*}

\noindent For the first term, by conditioning on the types of edge in $\ca{E_M}$:

\begin{align*}
\bb{I} \big\{ k_u \textrm{ all distinct}, \, (\dagger) \big\}
	&= \prod_{\ca{E}_\ca{M}^0} \bb{I} \{ k_{\sigma u} \neq k_{\sigma v} \} \\
	& \qquad \times \prod_{\ca{E}_\ca{M}^\mathrm{s}} \bb{I} \{ (k_{\sigma u}, k_{\sigma v}) \in \ca{E} \} \\
	& \qquad \times \prod_{\ca{E}_\ca{M}^\mathrm{d}} \bb{I} \{(k_{\sigma u}, k_{\sigma v}) \in \ca{E} \textrm{ and } (k_{\sigma v}, k_{\sigma u}) \in \ca{E}\} \\
&= \prod_{\ca{E}_\ca{M}^0} (J_\mathrm{n})_{k_{\sigma u},k_{\sigma v}}
	\prod_{\ca{E}_\ca{M}^\mathrm{s}} J_{k_{\sigma u},k_{\sigma v}}
	\prod_{\ca{E}_\ca{M}^\mathrm{d}} (J_\mathrm{d})_{k_{\sigma u},k_{\sigma v}} \\
&= J^\mathrm{func}_{\mathbf{k},\sigma}\,.
\end{align*}

\noindent Assuming $\big\{ k_u \textrm{ all distinct}, \, (\dagger) \big\}$, the second term is

\begin{align*}
\sum_{e \in \ca{E_H}} W(e)
	&= \sum_{\ca{E}_\ca{M}^\mathrm{s}} W((k_{\sigma u},k_{\sigma v}))
	+ \sum_{\ca{E}_\ca{M}^\mathrm{d}} \big( W((k_{\sigma u},k_{\sigma v})) + W((k_{\sigma v},k_{\sigma u})) \big) \\
&= \sum_{\ca{E}_\ca{M}^\mathrm{s}} G_{k_{\sigma u},k_{\sigma v}}
	+ \sum_{\ca{E}_\ca{M}^\mathrm{d}} (G_\mathrm{d})_{k_{\sigma u},k_{\sigma v}} \\
&= G^\mathrm{func}_{\mathbf{k},\sigma}
\end{align*}

\noindent as required.
For
\Cref{eq:mamformulastruc},
we simply change $(\dagger)$ to $(\ddagger)$ to check that an instance is a \emph{structural} instance

\[ (u,v) \in \ca{E_M} \iff (k_{\sigma u}, k_{\sigma v}) \in \ca{E} \qquad (\ddagger) \]

\pagebreak

\noindent Now for the first term, the following holds true

\begin{align*}
\bb{I} \big\{ k_u \textrm{ all distinct}, \, (\ddagger) \big\}
	&= \prod_{\ca{E}_\ca{M}^0} \bb{I} \{(k_{\sigma u}, k_{\sigma v}) \notin \ca{E} \textrm{ and } (k_{\sigma v}, k_{\sigma u}) \notin \ca{E}\} \\
	& \qquad \times \prod_{\ca{E}_\ca{M}^\mathrm{s}} \bb{I} \{(k_{\sigma u}, k_{\sigma v}) \in \ca{E} \textrm{ and } (k_{\sigma v}, k_{\sigma u}) \notin \ca{E}\} \\
	& \qquad \times \prod_{\ca{E}_\ca{M}^\mathrm{d}} \bb{I} \{(k_{\sigma u}, k_{\sigma v}) \in \ca{E} \textrm{ and } (k_{\sigma v}, k_{\sigma u}) \in \ca{E}\} \\
&= \prod_{\ca{E}_\ca{M}^0} (J_\mathrm{0})_{k_{\sigma u},k_{\sigma v}}
	\prod_{\ca{E}_\ca{M}^\mathrm{s}} (J_\mathrm{s})_{k_{\sigma u},k_{\sigma v}}
	\prod_{\ca{E}_\ca{M}^\mathrm{d}} (J_\mathrm{d})_{k_{\sigma u},k_{\sigma v}} \\
&= J^\mathrm{struc}_{\mathbf{k},\sigma}\,.
\end{align*}

\noindent Assuming $\big\{ k_u \textrm{ all distinct}, \, (\ddagger) \big\}$, the second term is given by

\begin{align*}
\sum_{e \in \ca{E_H}} W(e)
	&= \sum_{\ca{E}_\ca{M}^\mathrm{s}} W((k_{\sigma u},k_{\sigma v}))
	+ \sum_{\ca{E}_\ca{M}^\mathrm{d}} \big( W((k_{\sigma u},k_{\sigma v})) + W((k_{\sigma v},k_{\sigma u})) \big) \\
&= \sum_{\ca{E}_\ca{M}^\mathrm{s}} (G_\mathrm{s})_{k_{\sigma u},k_{\sigma v}}
	+ \sum_{\ca{E}_\ca{M}^\mathrm{d}} (G_\mathrm{d})_{k_{\sigma u},k_{\sigma v}} \\
&= G^\mathrm{struc}_{\mathbf{k},\sigma}\,.
\end{align*}

\hfill $\square$
\end{prf}

\begin{prf}[\Cref{prop:coll_expa_formulae}, Colliders and expanders in bipartite graphs] \label{proof:coll_expa_formulae}

Consider
\Cref{eq:mamcoll}
and the collider motif $\ca{M}_\mathrm{coll}$. Since $\ca{G}$ is bipartite, $M_\mathrm{coll}^\mathrm{func} = M_\mathrm{coll}^\mathrm{struc} = \vcentcolon M_\mathrm{coll}$, and by \Cref{tab:motif_adj_mat_table},  $M_\mathrm{coll} = \frac{1}{2} J_\mathrm{n} \circ (J G^\mathsf{T} + G J^\mathsf{T})$. Hence

\vspace{2mm}
\begin{align*}
	(M_\mathrm{coll})_{ij} &= \frac{1}{2} (J_\mathrm{n})_{ij} \ (J G^\mathsf{T} + G J^\mathsf{T})_{ij} \\
	&= \bb{I}\{i \neq j\} \sum_{k \in \ca{V}} \ \frac{1}{2} \Big(J_{ik} G_{jk} + G_{ik} J_{jk} \Big) \\
	&= \bb{I}\{i \neq j\} \sum_{k \in \ca{V}} \ \frac{1}{2} \,\bb{I} \, \Big\{ (i,k),(j,k) \in \ca{E} \Big\} \Big[W((i,k)) + W((j,k))\Big] \\
	&= \bb{I} \{i \neq j\} \hspace*{-0.4cm} \sum_{\substack{k \in \ca{D} \\ (i,k), (j,k) \in \ca{E}}} \hspace*{-0.2cm} \frac{1}{2} \Big[ W((i,k)) + W((j,k)) \Big]\,.
\end{align*}

\noindent Similarly for the expander motif $M_\mathrm{expa} = \frac{1}{2} J_\mathrm{n} \circ (J^\mathsf{T} G + G^\mathsf{T} J)$, which yields Equation (4):

\begin{align*}
	(M_\mathrm{expa})_{ij} &= \frac{1}{2} (J_\mathrm{n})_{ij} \ (J^\mathsf{T} G + G^\mathsf{T} J)_{ij} \\
	&= \bb{I} \{i \neq j\} \hspace*{-0.4cm} \sum_{\substack{k \in \ca{S} \\ (k,i), (k,j) \in \ca{E}}} \hspace*{-0.2cm} \frac{1}{2} \Big[ W((k,i)) + W((k,j)) \Big]\,.
\end{align*}

\hfill $\square$
\end{prf}

\pagebreak

\begin{prf}[\Cref{prop:motif_adj_matrix_computation}, Complexity of MAM formula] \label{proof:motif_adj_matrix_computation}
Suppose ${m \leq 3}$ and consider $M^\mathrm{func}$. The adjacency and indicator matrices of $\ca{G}$ are given by

\vspace{2mm}
\begin{equation*}
	\begin{aligned}[c]
		&(1) \quad J = \bb{I} \{ G>0 \}\,, \\
		&(2) \quad J_0 = \bb{I} \{ G + G^\mathsf{T} = 0 \} \circ J_\mathrm{n}\,, \\
		&(3) \quad J_\mathrm{s} = J - J_\mathrm{d}\,, \\
		&(4) \quad G_\mathrm{d} = (G + G^\mathsf{T}) \circ J_\mathrm{d} \,,
	\end{aligned}
	\hspace*{2cm}
	\begin{aligned}[c]
		&(5) \quad J_\mathrm{n} = \bb{I} \{I_{n \times n} = 0 \}\,, \\
		&(6) \quad J_\mathrm{d} = J \circ J^\mathsf{T}\,, \\
		&(7) \quad G_\mathrm{s} = G \circ J_\mathrm{s}\,, \\
		&
	\end{aligned}
\end{equation*}
\vspace{1mm}

\noindent and computed using four additions and four element-wise multiplications. $J^\mathrm{func}_{\mathbf{k},\sigma}$ is a product of at most three factors, and $G^\mathrm{func}_{\mathbf{k},\sigma}$ contains at most three summands, thus

\vspace{1mm}
\[ \sum_{k_2 \in \ca{V}} J^\mathrm{func}_{\mathbf{k},\sigma} \ G^\mathrm{func}_{\mathbf{k},\sigma} \]
\vspace{1mm}

\noindent is expressible as a sum of at most three matrices, each of which is constructed with at most one matrix multiplication (where $\{k_{\sigma r},k_{\sigma s}\} \neq \{i,j\}$) and one entry-wise multiplication (where $\{k_{\sigma r},k_{\sigma s}\} = \{i,j\}$). This is repeated for each $\sigma \in S_\ca{M,A}^\sim$ (at most six times) and the results are summed. Calculations are identical for  $M^\mathrm{struc}$.

\hfill $\square$
\end{prf}

 \section{Spectral Clustering} \label{chap:spectral}

We provide a summary of traditional random-walk spectral clustering and show how it applies to motif-based clustering.

\subsection{Overview of spectral clustering} \label{sec:spectral_overview}
For an undirected graph, the objects to be clustered are the vertices of the graph, and a similarity matrix is provided by the graph's adjacency matrix $G$.
To cluster directed graphs, the adjacency matrix must first be symmetrized,
perhaps
by considering
$G + G^\mathsf{T}$
\citep{Meila2007ClusteringBW}.
This symmetrization ignores information about edge direction and higher-order structures, and can lead to poor performance
\citep{benson2016higher}.
Spectral clustering consists of two steps. Firstly, eigendecomposition of a Laplacian matrix embeds the vertices into $\bb{R}^{l}$. The $k$ clusters are then extracted from this space.

\subsection{Graph Laplacians} \label{sec:appendixlaplacians}

The Laplacians of an undirected graph are a family of matrices which play a central role in spectral clustering.
While many different graph Laplacians are available,
we focus
on just the \emph{random-walk Laplacian},
for reasons concerning objective functions,
consistency and computation \citep{luxburg2004convergence, von2007tutorial}.
The random-walk Laplacian of an undirected graph $\ca{G}$ with
(symmetric) adjacency matrix $G$ is given by

\[ L_\mathrm{rw} \vcentcolon= I - D^{-1} G, \]

\noindent where $I$ is the identity and $D_{ii} \vcentcolon= \sum_j G_{ij}$ is the diagonal matrix of weighted degrees.

\subsection{Graph cuts} \label{sec:spectral_graph_cut}
Graph cuts provide objective functions which we seek to minimize while clustering the vertices of a graph.

\begin{definition}
Let $\ca{G}$ be a graph. Let $ \ca{P}_1, \ldots, \ca{P}_k $ be a partition of $\ca{V}$. Then the \emph{normalized cut} \citep{shi2000normalized} of $\ca{G}$ with respect to $ \ca{P}_1, \ldots, \ca{P}_k $ is

$$ \mathrm{Ncut}_\ca{G}(\ca{P}_1, \ldots, \ca{P}_k) \vcentcolon= \frac{1}{2} \sum_{i=1}^k \frac{ \mathrm{cut}(\ca{P}_i,\bar{\ca{P}_i}) }{ \mathrm{vol}(\ca{P}_i) }, $$

\vspace{1mm}
\noindent where $ \mathrm{cut}(\ca{P}_i,\bar{\ca{P}_i}) \vcentcolon= \sum_{u \in \ca{P}_i, \, v \in \ca{V} \setminus \ca{P}_i} G_{uv}$ and $\mathrm{vol}(\ca{P}_i) \vcentcolon= \sum_{u \in \ca{P}_i} D_{uu}$.
\end{definition}
\noindent Note that more desirable partitions have a lower Ncut value;
the numerators penalize partitions which cut a large number of heavily weighted edges,
and the denominators penalize partitions which have highly imbalanced cluster sizes.
It can be shown \citep{von2007tutorial} that minimizing Ncut over partitions $ \ca{P}_1, \ldots, \ca{P}_k $ is equivalent to finding the cluster indicator matrix $H \in \bb{R}^{n \times k}$ minimizing
$ \mathrm{Tr} \big( H^\mathsf{T} (D-G) H \big) $
subject to
$ H_{ij} = \mathrm{vol}(\ca{P}_j)^{-\frac{1}{2}} \ \bb{I} \{ v_i \in \ca{P}_j \}\, (\dagger) $, and  $ H^\mathsf{T} D H = I\,$.
Solving this problem is in general \textsf{NP}-hard \citep{wagner1993between}.
However, by dropping the constraint~$(\dagger)$ and applying the Rayleigh Principle \citep{lutkepohl1996handbook}, we find that the solution to this relaxed problem is that $H$ contains the first $k$ eigenvectors of $L_\mathrm{rw}$ as columns \citep{von2007tutorial}.
In practice, to find $k$ clusters it is often sufficient to use only the first $l < k$ eigenvectors of $L_\mathrm{rw}$.

\subsection{Cut imbalance ratio}
\label{sec:cut_imbalance_ratio}

\green{Given a partition of a directed graph,
cut imbalance ratio gives a notion of
flow imbalance between each pair of clusters.}

\begin{definition}

\green{Let $\ca{G}$ be a graph. Let $ \ca{P}_1, \ldots, \ca{P}_k $ be a partition of $\ca{V}$.
Then the \emph{cut imbalance ratio} \citep{DirectedClustImbCuts}
of $\ca{G}$ from $\ca{P}_i$ to $\ca{P}_j$
is}

$$ \mathrm{CIR}_\ca{G}(\ca{P}_i, \ca{P}_j) \vcentcolon= \frac{1}{2} \frac{ \mathrm{cut}(\ca{P}_i,\ca{P}_j) - \mathrm{cut}(\ca{P}_j, \ca{P}_i)}{ \mathrm{cut}(\ca{P}_i,\ca{P}_j) + \mathrm{cut}(\ca{P}_j, \ca{P}_i)},$$

\vspace{1mm}
\noindent \green{where $ \mathrm{cut}(\ca{P}_i,\ca{P}_j) \vcentcolon= \sum_{u \in \ca{P}_i, \, v \in \ca{P}_j} G_{uv}$.}
\end{definition}
\noindent \green{The cut imbalance ratio takes values in
$[-\frac{1}{2}, \frac{1}{2}]$,
with positive values indicating more flow from
$\ca{P}_i$ to $\ca{P}_j$
and negative values indicating more flow from
$\ca{P}_j$ to $\ca{P}_i$.
The \emph{cut imbalance ratio matrix}
is the antisymmetric $k \times k$ matrix with
$(i,j)$th entry equal to
$\mathrm{CIR}_\ca{G}(\ca{P}_i, \ca{P}_j)$.}

\subsection{Cluster extraction} \label{sec:spectral_cluster_extraction}
Once Laplacian eigendecomposition has been used to embed the data into $\bb{R}^l$, the clusters may be extracted using a variety of methods.
We propose $k$-means++ \citep{arthur2007k}, a popular clustering algorithm for data in $\bb{R}^l$,
as an appropriate technique.
It aims to minimize the within-cluster sum of squares, based on the standard Euclidean metric on $\bb{R}^l$.
This makes it a reasonable candidate for clustering spectral data, since the Euclidean metric corresponds to notions of ``diffusion distance'' in the original graph \citep{nadler2006diffusion}.

\subsection{Random-walk spectral clustering}

Algorithm~\ref{alg:rwspectclust} gives random-walk spectral clustering \citep{von2007tutorial},
which takes a symmetric connected adjacency matrix as input.
We drop the first column of $H$ (the first eigenvector of $L_\mathrm{rw}$)
since although it should be constant and uninformative,
numerical imprecision may give unwanted artifacts.
It is worth noting that although the relaxation used in \Cref{sec:spectral_graph_cut} is reasonable and often leads to good approximate solutions of the Ncut problem,
there are cases where it performs poorly~\citep{guattery1998quality}.
The Cheeger inequality~\citep{chung2005laplacians} gives a bound on the error introduced by this relaxation.

\subsection{Motif-based random-walk spectral clustering} \label{sec:motifrwspectclust}

Algorithm~\ref{alg:motifrwspectclust} gives motif-based random-walk spectral clustering.
The algorithm forms a motif adjacency matrix,
restricts it to its largest connected component,
and applies random-walk spectral clustering
(\Cref{alg:rwspectclust})
to produce clusters.
The most computationally expensive part of Algorithm~\ref{alg:motifrwspectclust}
is the calculation of the MAM using a formula from Table~\ref{tab:motif_adj_mat_table},
as noted by
\cite{benson2016higher}
for their unweighted MAMs.
The complexity of this is analyzed in
\Cref{sec:computational_analysis}.

\subsection{Bipartite spectral clustering}

Algorithm~\ref{alg:bipartite_clustering} gives our procedure for clustering a bipartite graph.
The algorithm uses the collider and expander motifs to create similarity matrices
for the source and destination vertices respectively
(as in \Cref{sec:bipartite}),
and then applies random-walk spectral clustering (Algorithm~\ref{alg:rwspectclust}) to produce the partitions.

\vspace{3mm}
\begin{algorithm}[H]
  \caption{Random-walk spectral clustering}
  \label{alg:rwspectclust}

	\SetKwFunction{Main}{RWSpectClust}				\newcommand{\MainArgs}{$G,k,l$}		

 	\BlankLine
	\Input{Symmetric adjacency matrix $G$, number of clusters $k$, dimension $l$}
	\Output{Partition $\ca{P}_1, \ldots, \ca{P}_k$}
	\BlankLine

  \vspace{-0.1cm}
	\Function{\Main{\MainArgs}}{
		Construct the weighted degree matrix $D_{ii} \leftarrow \sum_j G_{ij}$ \\
		Construct the random-walk Laplacian matrix $L_\mathrm{rw} \leftarrow I-D^{-1}G$ \\
		Let $H$ have the first $l$ eigenvectors of $L_\mathrm{rw}$ as columns \\
		Drop the first column of $H$ \\
		Run $k$-means++ on the rows of $H$ with $k$ clusters to produce $\ca{P}_1, \ldots, \ca{P}_k$ \\
    \Return	$\ca{P}_1, \ldots, \ca{P}_k$
	}

\end{algorithm}

\begin{algorithm}[H]
  \caption{Motif-based random-walk spectral clustering}
  \label{alg:motifrwspectclust}

	\SetKwFunction{Main}{MotifRWSpectClust}				\newcommand{\MainArgs}{$\ca{G},\mathcal{M},k,l$}		

 	\BlankLine
	\Input{Graph $\ca{G}$, motif $\ca{M}$, number of clusters $k$, dimension $l$}
	\Output{Partition $\ca{P}_1, \ldots, \ca{P}_k$}
	\BlankLine

  \vspace{-0.1cm}
	\Function{\Main{\MainArgs}}{
		Construct the motif adjacency matrix $M$ of the graph $\ca{G}$ with  motif $\ca{M}$ \\
		Let $\tilde{M}$ be $M$ restricted to its largest connected component, $C$ \\
		$\ca{P}_1, \ldots, \ca{P}_k \leftarrow$ \texttt{RWSpectClust($\tilde{M},k,l$)} \\
    \Return	$\ca{P}_1, \ldots, \ca{P}_k$
	}

\end{algorithm}

\vspace{3mm}
\begin{algorithm}[H]

	\SetKwFunction{Main}{BipartiteRWSpectClust}				\newcommand{\MainArgs}{$\ca{G},k_\ca{S},k_\ca{D},l_\ca{S},l_\ca{D}$}		

 	\BlankLine
	\Input{Bipartite graph $\ca{G}$, source clusters $k_\ca{S}$, destination clusters $k_\ca{D}$, source dimension $l_\ca{S}$, destination dimension $l_\ca{D}$}
	\Output{Source partition $\ca{S}_1, \ldots, \ca{S}_{k_\ca{S}}$, destination partition $\ca{D}_1, \ldots, \ca{D}_{k_\ca{D}}$}
	\BlankLine
\vspace{-0.1cm}
	\Function{\Main{\MainArgs}}{
		Construct the collider motif adjacency matrix $M_\mathrm{coll}$ of the graph $\ca{G}$ \\
		Construct the expander motif adjacency matrix $M_\mathrm{expa}$ of the graph $\ca{G}$ \\
		$M_\mathrm{coll} \leftarrow M_\mathrm{coll}[\ca{S,S}]$ \Comm*{restrict rows and columns of $M_\mathrm{coll}$ to $\ca{S}$ \hspace*{0.07cm}}
		$M_\mathrm{expa} \leftarrow M_\mathrm{expa}[\ca{D,D}]$ \Comm*{restrict rows and columns of $M_\mathrm{expa}$ to $\ca{D}$}
		$\ca{S}_1, \ldots, \ca{S}_{k_\ca{S}} \leftarrow$ \texttt{RWSpectClust($M_\mathrm{coll},k_\ca{S},l_\ca{S}$)} \\
		$\ca{D}_1, \ldots, \ca{D}_{k_\ca{D}} \leftarrow$ \texttt{RWSpectClust($M_\mathrm{expa},k_\ca{D},l_\ca{D}$)} \\
	\Return	$\ca{S}_1, \ldots, \ca{S}_{k_\ca{S}}$ and  $\ca{D}_1, \ldots, \ca{D}_{k_\ca{D}}$
	}

	\caption{Bipartite random-walk spectral clustering}
	\label{alg:bipartite_clustering}
\end{algorithm}

\vfill
\pagebreak

\section{Additional Synthetic Results}
\label{sec:additionSyn}

In this section we present additional results
for some of our synthetic experiments.

\paragraph{Example 1}
In \Cref{fig:weightedStruc} we present the results for structural motifs,
to complement the functional motifs presented in the main paper.
We note that when we consider structural motifs,
the unweighted motifs perform better
but are still outperformed by our weighted method.

\paragraph{Example 2 -- Experiment 1}
In \Cref{fig:benchmark1FUNC}
we present the results for the functional motifs for Example 2 Experiment 1. As noted in the main paper, our approaches outperform the baseline, but we do not observe the same bi-modal structure, with $\mathcal{M}_1$ outperforming all approaches for $q_1<0.6$, and without the recovery in performance for high values of $q_1$.

\paragraph{Example 2 -- Experiment 2}
In \Cref{fig:benchmark1FigStruc} we present the structural version of Example 2 Experiment 2. We observe the same structures as we observe for the functional version, with motifs each highlighting a different but equally relevant structure. Finally the symmetrized adjacency matrix ($\mathcal{M}_s$),
for which each of the blocks are indistinguishable, does not robustly identify the higher-order structures for $k=2$, but can uncover the blocks when clustering with $k=8$ and $l=3$.

\wf{\begin{figure}[h!]\centering
    \includegraphics[width=10cm]{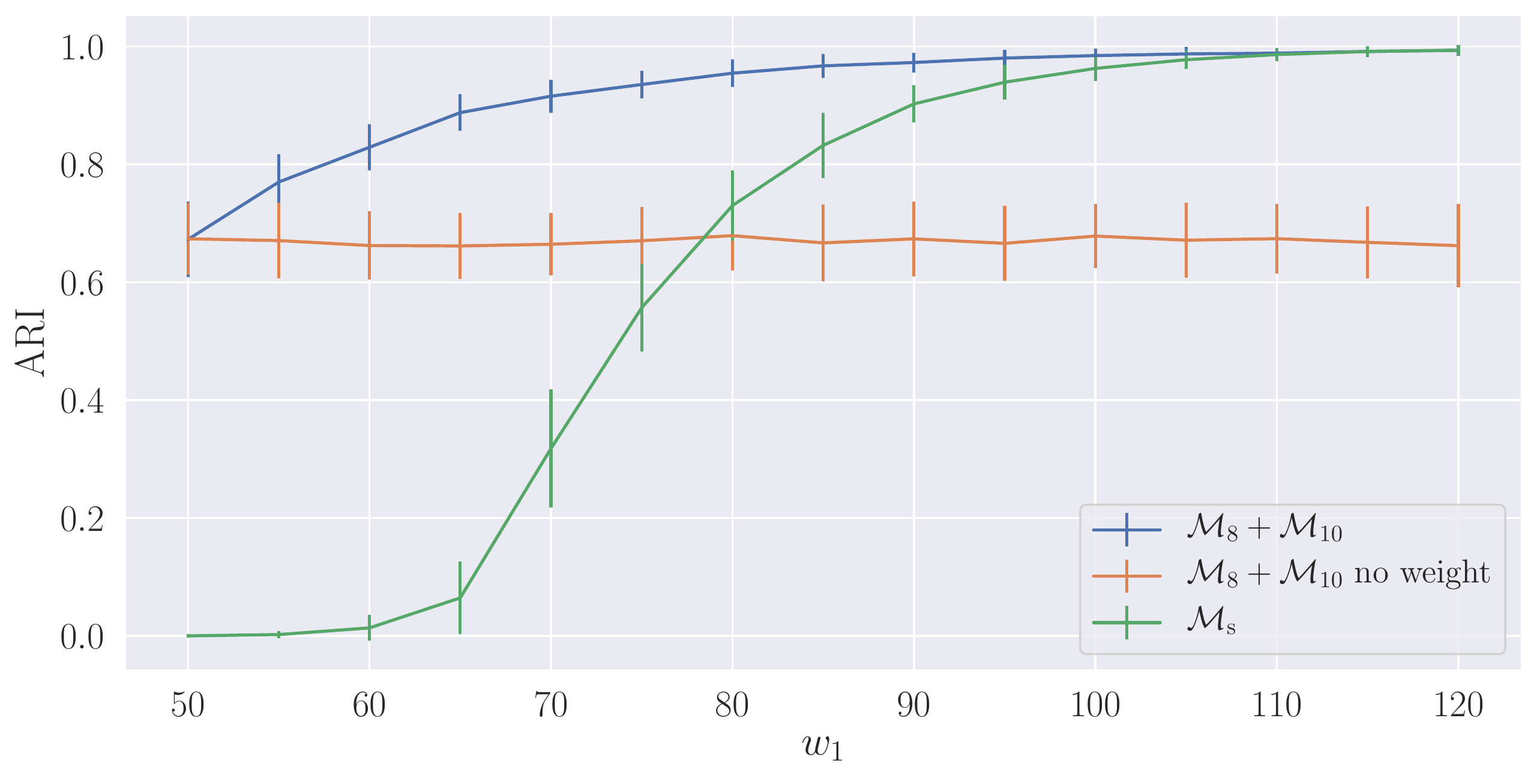}
	\caption{
	Exploring the performance of structural MAMs based on $\mathcal{M}_8$ and $\mathcal{M}_{10}$ on
    Example 1
(see \Cref{fig:weighted}).
	Each block contains $100$ nodes. We compare both to the unweighted case, and to the symmetrized case. We perform $100$ repeats, and error bars are one sample standard deviation.
	}
	\label{fig:weightedStruc}
\end{figure}
}{}

\wf{\begin{figure}[h!]\centering
    \includegraphics[width=10cm]{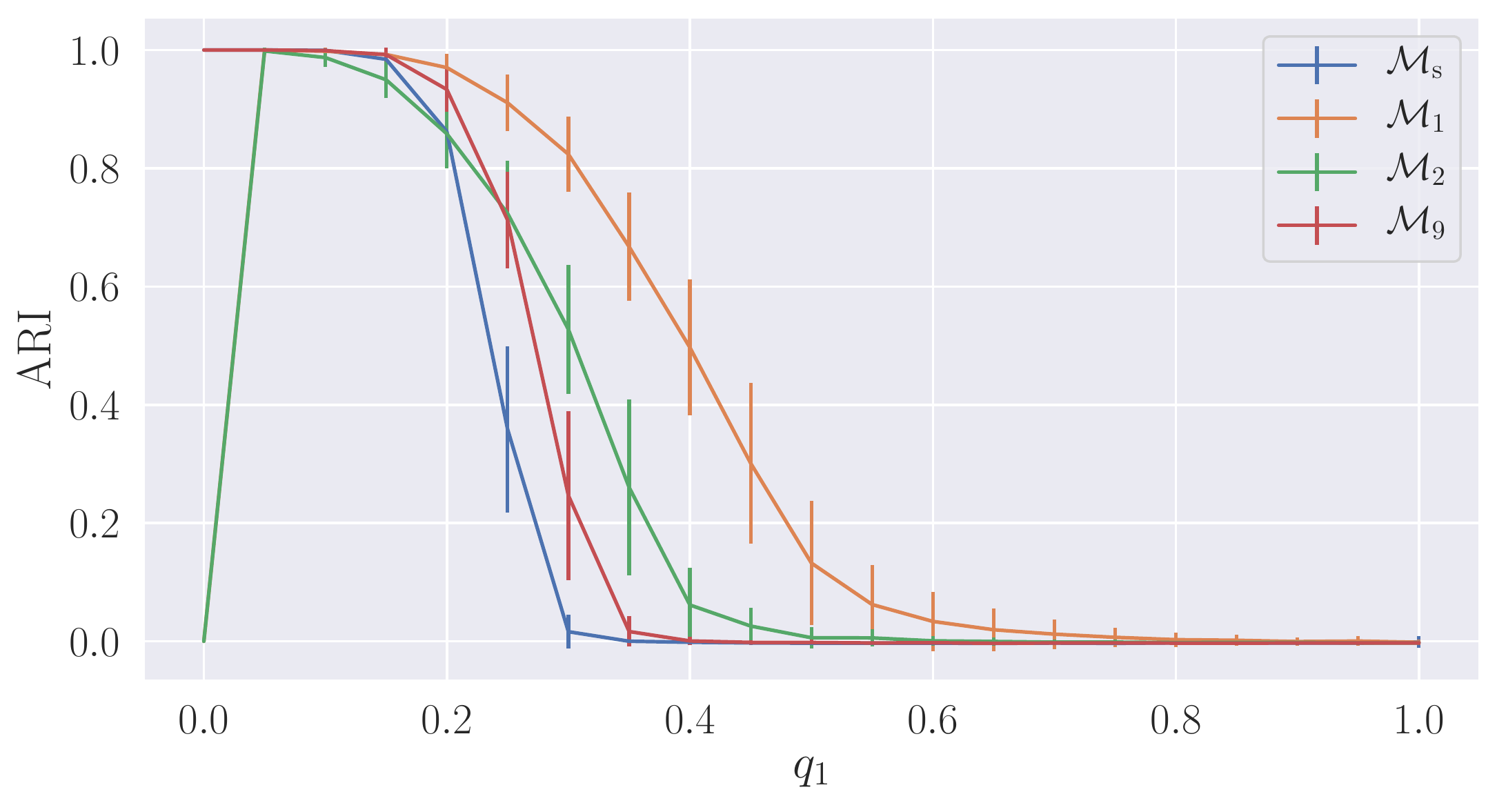}
    \caption{Performance on Example 2 Experiment 1 using functional MAMs.
    We compare with the standard symmetrization $\mathcal{M}_s$.
    We perform
    $100$
    repeats, and error bars are one sample standard deviation.
    }
  \label{fig:benchmark1FUNC}
\end{figure}
}{}

\wf{\begin{figure}[h!]\centering
  \includegraphics[width=10cm]{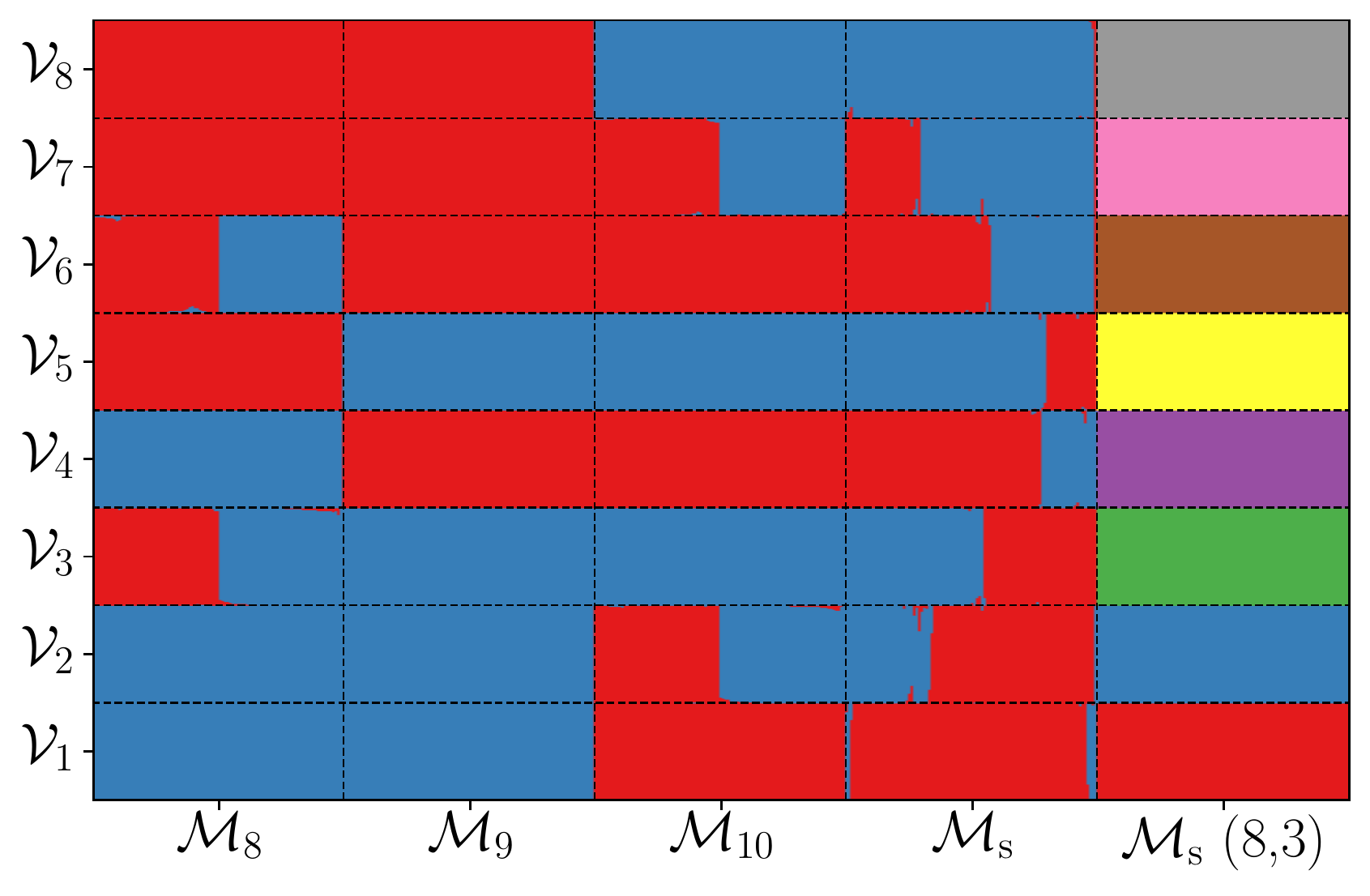}
\caption{
The detected groups uncovered  by each method using structural motifs, with $k=2$ and $l=2$, with the exception of the last column which has $k=8$, $l=3$.
We test $100$ replicates and present the results as columns in the plot.
We order and color the columns as in
\Cref{fig:benchmark1Fig}.
While the clustering assignments may contain some errors,
the results are relatively robust across replicates.
}
\label{fig:benchmark1FigStruc}
\end{figure}
}{}

\section{Additional Real World Results}
\label{sec:additionReal}

\green{In this section we present additional results
for some of our real world experiments.}

\paragraph{\textsc{US-Migration} network}
\green{\Cref{fig:appendix_migration_clusts_func}
plots clusterings obtained from the
\textsc{US-Migration} network using all 15 weighted
functional motifs.
Note how the different motifs uncover a variety of different clusterings,
and also that they all display good spatial coherence.
For completeness,
\Cref{fig:appendix_migration_clusts_struc}
shows the equivalent plots for structural motifs.
Note that motifs $\ca{M}_\mathrm{d}$
and $\ca{M}_4$ are bi-directional cliques on two and
three vertices respectively,
so their functional and structural MAMs are identical.
The spatial coherence deteriorates for certain structural motifs,
and cluster sizes are occasionally less balanced.
We believe this is to do with the high local density of the migration network,
preventing some motifs from appearing as structural instances
in certain places
(see \Cref{sec:syntheticexample2} for another
example of structural motifs exhibiting this
filtering property,
which may or may not be desirable,
depending on the application).
Note that this effect is most pronounced for
the motifs containing single edges
(i.e. excluding $\ca{M}_\mathrm{d}$, $\ca{M}_4$ and $\ca{M}_{13}$),
because structural instances of single edges require an edge in precisely
one direction,
filtering out pairs of nodes which exhibit either no mutual migration
or bi-directional mutual migration.}

\wf{\begin{figure}[t]
  \centering
  \includegraphics[width=0.95\textwidth]{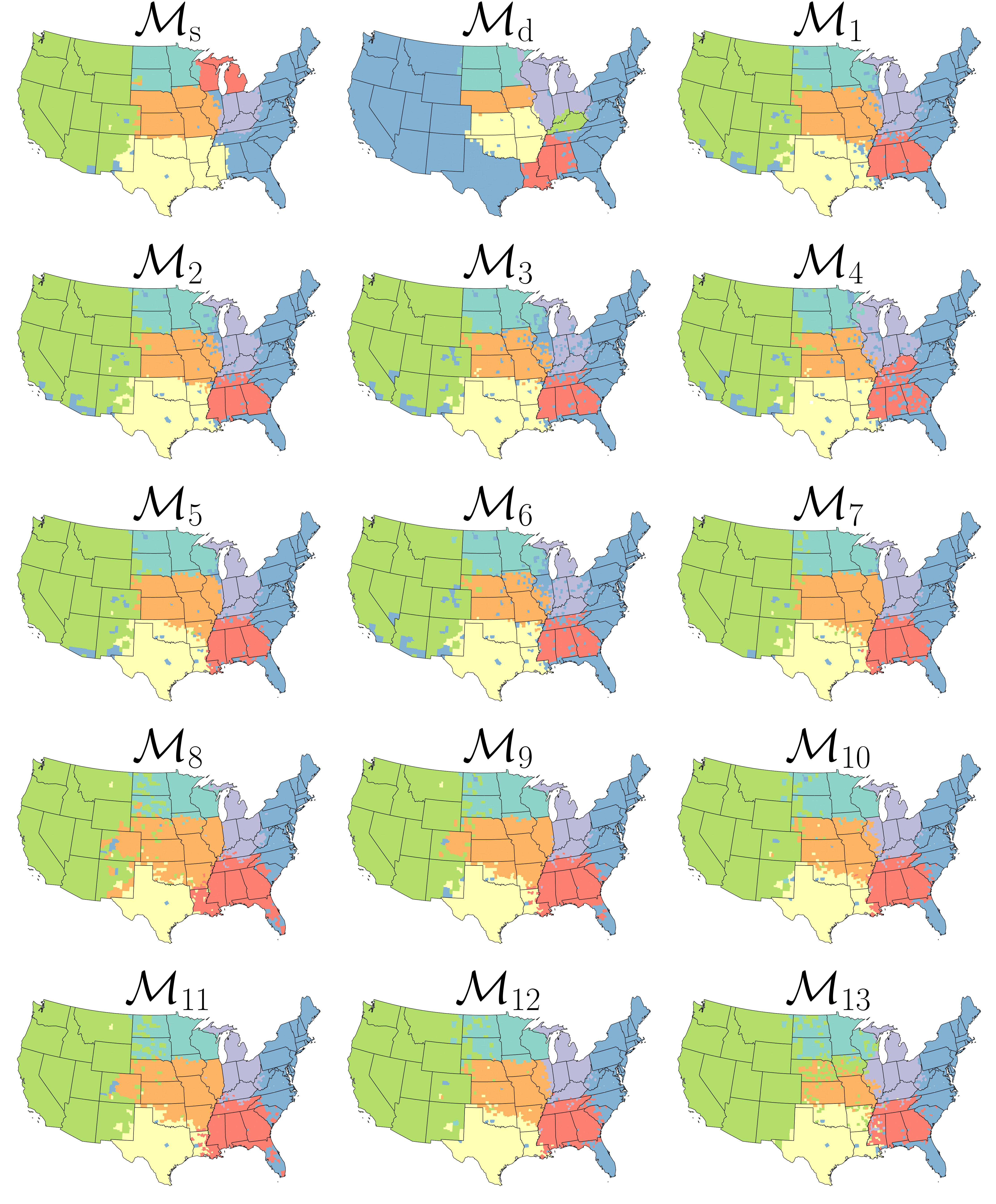}
\caption{\green{Weighted functional motif-based clusterings of the \textsc{US-Migration} network, for all motifs on at most three nodes.
  Each county (node) is colored by its cluster allocation.
  }}
\label{fig:appendix_migration_clusts_func}
\end{figure}
}{}

\wf{\begin{figure}[t]
  \centering
  \includegraphics[width=0.95\textwidth]{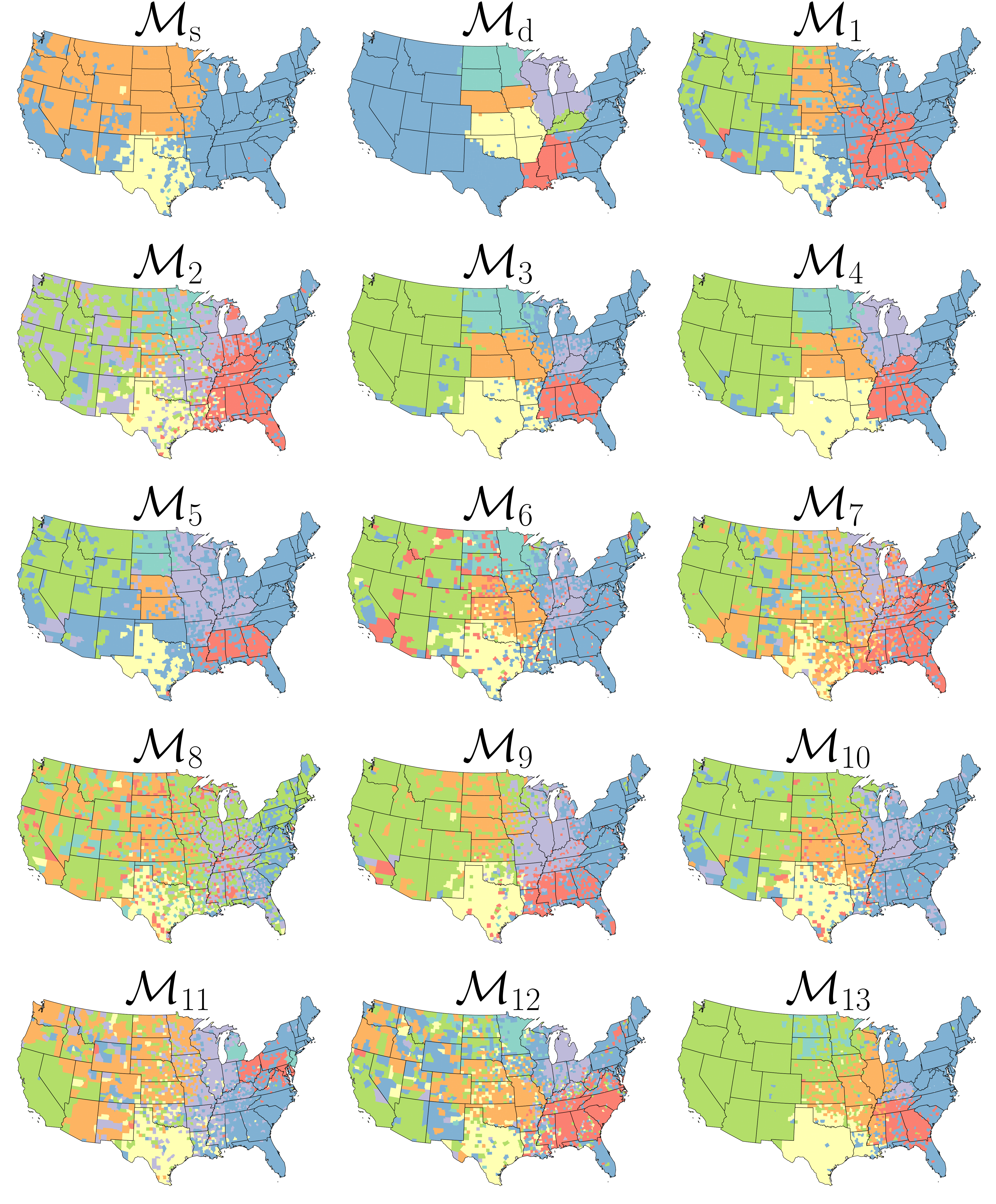}
\caption{\green{Weighted structural motif-based clusterings of the \textsc{US-Migration} network, for all motifs on at most three nodes.
  Each county (node) is colored by its cluster allocation.
  }}
\label{fig:appendix_migration_clusts_struc}
\end{figure}
}{}

\end{document}